\documentclass[12pt]{article}

\usepackage{amssymb}
\usepackage{amsmath}
\usepackage{amsthm}
\usepackage{mathrsfs}
\usepackage{MnSymbol}
\usepackage{cite}
\usepackage{graphicx}
\usepackage{appendix}
\usepackage{enumerate}
\usepackage{float}

\usepackage{calligra}
\DeclareMathAlphabet{\mathcalligra}{T1}{calligra}{m}{n}
\DeclareFontShape{T1}{calligra}{m}{n}{<->s*[2.2]callig15}{}

\usepackage[cal=boondox]{mathalfa}

\newcommand{\R}{\ensuremath{\mathbb{R}}}

\newcommand{\RN}{\ensuremath{\mathbb{R}^N}}
\newcommand{\scrp}{\ensuremath{\mathcal{p}}}
\newcommand{\scrq}{\ensuremath{\mathcal{q}}}
\newcommand{\scrm}{\ensuremath{\mathcal{m}}}

\newcommand{\scrv}{\ensuremath{\mathcal{v}}}
\newcommand{\scrw}{\ensuremath{\mathcal{w}}}

\newcommand{\scrP}{\ensuremath{\mathscr{P}}}
\newcommand{\scrC}{\ensuremath{\mathscr{C}}}
\newcommand{\scrPhat}{\ensuremath{\hat{\mathscr{P}}}}
\newcommand{\scrChat}{\ensuremath{\hat{\mathscr{C}}}}

\newcommand{\scrV}{\ensuremath{\mathscr{V}}}
\newcommand{\scrM}{\ensuremath{\mathscr{M}}}

\newcommand{\supp}{\ensuremath{\mathrm{supp} \,}}

\newcommand{\Cone}{\ensuremath{\mathrm{Cone} \,}}

\newcommand{\theory}{\ensuremath{(\Sigma,\scrP)}}
\newcommand{\deltam}{\ensuremath{\Delta\scrm}}
\newcommand{\process}{\ensuremath{(\Delta\scrm,\scrq)}}
\newcommand{\MSigma}{\ensuremath{\scrM(\Sigma)}}
\newcommand{\MSigmaZ}{\ensuremath{\scrM^{\circ}(\Sigma)}}
\newcommand{\MSigmaPl}{\ensuremath{\scrM_{+}(\Sigma)}}
\newcommand{\MSigmaPlOne}{\ensuremath{\scrM_{+}^1(\Sigma)}}
\newcommand{\VSigma}{\ensuremath{\scrV(\Sigma)}}
\newcommand{\etao}{\ensuremath{\eta^{\circ}}}
\newcommand{\To}{\ensuremath{T^{\circ}}}

\newtheorem{theorem}{\bf Theorem}[section]

\newtheorem{corollary}{\bf Corollary}[section]

\theoremstyle{remark}
\newtheorem{rem}[theorem]{Remark}
\newtheorem{example}[theorem]{Example}

\begin{document}

\title{How the Hahn-Banach Theorem Sheds Bright Light on Fundamental Questions in Classical Thermodynamics}

\author{Martin Feinberg\thanks{The William G. Lowrie Department of Chemical \& Biomolecular Engineering and Department of Mathematics, The Ohio State University, 151 W. Woodruff Avenue, Columbus, OH 43210 USA.  E-mail: feinberg.14@osu.edu.}
\and Richard B. Lavine\thanks{Department of Mathematics, University of Rochester, Rochester, NY 14627 USA.  Email: rdlavine@frontiernet.net}}






\maketitle

\begin{abstract}
The Hahn-Banach Theorem, a cornerstone of modern functional analysis,  is a natural companion of the Second Law of Thermodynamics. From a Kelvin-Planck version of the Second Law, the Hahn-Banach Theorem delivers, immediately and simultaneously, entropy and thermodynamic-temperature functions of the local material state such that the Clausius-Duhem inequality is satisfied for every process a particular material might admit. For \emph{existence} of such functions there is no need at all to require that their domain be restricted to states of equilibrium. However, the Hahn-Banach Theorem also indicates that for \emph{uniqueness} of such a pair of functions across the entire state-space domain, every state must be visited by a reversible process. This review is intended to help make accessible to both thermodynamics scholars and  mathematicians the remarkable interplay of the Hahn-Banach Theorem and the Second Law.
\end{abstract}
\newpage


\section{Introduction}
 
The last of the iconic 19\textsuperscript{th} century papers on foundations of classical thermodynamics, Gibbs's \emph{On the Equilibrium of Heterogeneous Substances} \cite{gibbs_ehs}, appeared in 1875. That was the year in which the great mathematician Henri Lebesgue was born. Without Lebesgue, functional analysis --- a core subject of modern mathematics --- would not have evolved to what it is today. 

By an accident of timing, then, the brilliant thermodynamics pioneers --- Carnot \cite{carnot1824}, Clapeyron \cite{clapeyron1834memoir}, Clausius \cite{clausius1854veranderte,ClausiusBook1867}, Kelvin \cite{thomson1853xv}, and Gibbs --- did not have available to them the powerful tools that modern functional analysis provides, or even the language required for precise mathematical expression of their ideas. In many ways, classical thermodynamics remains firmly  rooted in the 19\textsuperscript{th} century, with both teaching and research largely divorced from powerful benefits readily derived from the framework and theorems of 20\textsuperscript{th} century functional analysis.\footnote{Standard presentations of classical thermodynamics are well represented in 20\textsuperscript{th} century textbooks by Fermi \cite{fermithermodynamics}, Pippard \cite{pippard1964elements} and Denbigh \cite{denbigh1981principles}. More technical versions of the Hahn-Banach Theorem do play (less starring) roles in \cite{giles2016mathematical} and, much later, in \cite{lieb1999physics}.} In the 21\textsuperscript{st} century, this shouldn't be the case.

It is the purpose of this article to make more accessible, to thermodynamics scholars and to mathematicians\footnote{There is a growing influence of thermodynamics on mathematics. See, for example, surveys of entropy ideas in partial differential equations \cite{evans2004entropy,ball2013entropy}.}, salient results \cite{feinberg1983thermodynamics,feinberg1986foundations,feinberglavine2024entropy1,feinberglavine2024entropy2} dating back to the 1980s  at the juncture of the two subjects. That work indicates how modern functional analysis, via the Hahn-Banach Theorem (explained below), delivers \emph{simultaneously} both entropy and thermodynamic temperature as \emph{immediate} consequences of the Second Law.  Because emergence of entropy and thermodynamic temperature, as functions of state, becomes disentangled from Carnot-cycle near-equilibrium arguments of the pioneers, the Hahn-Banach viewpoint provides means to examine the 19\textsuperscript{th} century arguments critically and to shed light on questions arising from them --- \emph{in particular about the extent to which domains of the entropy and thermodynamic temperature functions are necessarily restricted to equilibrium (or near-equilibrium) states}.

\subsection{The Legacy of the Thermodynamics Pioneers}\label{subsec:Legacy} The 19\textsuperscript{th} century pioneers did what they could with what they had, and they did it brilliantly. Beginning with one or another (necessarily vague) word statement of the Second Law --- a stricture on the nature of net heat transfer to  a body experiencing a cyclic process --- they deduced from ingenious arguments, involving coupled, highly idealized engines and refrigerators, the existence of two functions of a material's state --- an \emph{entropy function} and a \emph{thermodynamic temperature function} --- that, taken together, satisfy a large family of \emph{linear integral inequalities} constraining those two functions.

In particular, from the Second Law the pioneers surmised the existence of entropy and thermodynamic temperature functions of the material state that, \emph{for each process}, satisfy what has come to be called the \emph{Clausius-Duhem inequality}. In the very first paragraph of  \emph{Equilibrium of Heterogeneous Substances}, Gibbs expressed that requirement in the following way: For each process,

\begin{equation}
\label{eq:CDInGibbs}
\begin{bmatrix}
\textrm{The total entropy}\\  
\textrm{of the body at the}\\  
\textrm{end of the process}
\end{bmatrix}
\,-\,
\begin{bmatrix}
\textrm{The total entropy of the}\\   
\textrm{body at the beginning}\\  
 \textrm{of the process}
\end{bmatrix}
\quad\geq\quad
{\int\frac{dq}{T}}\ \biggr\rvert_{\;\textrm{process}},
\end{equation}

\noindent
\emph{``dq denoting the element of heat received from external sources and $T$ denoting the temperature of the part of the system receiving it."} (This interpretation of the right side of \eqref{eq:CDInGibbs} is from the same Gibbs paragraph.) 

That each term on the left side of \eqref{eq:CDInGibbs} is also an integral is made clear in a passage from an earlier article by Gibbs \cite{gibbs_geom}. There he states clearly that, even for an \emph{unequilibrated} body, 
\emph{``The body, however, as a whole has a certain volume, entropy, and energy which are equal to the sums of the volumes, etc., of its parts."} This suggests that, if local entropy values can be calculated from a function of the local material state, then the entropy of a body as a whole can be calculated, at each instant, from an integration of the local entropy function over the entire body, taking account of the distribution within the body of the local material states. Note that such an integral is \emph{linear} in the local entropy \emph{function}, just as the right side of \eqref{eq:CDInGibbs} is \emph{linear} in the reciprocal of the local thermodynamic temperature \emph{function}. 

In a sense, then, the infinite set of constraints given by \eqref{eq:CDInGibbs} on the entropy and the reciprocal-temperature functions --- that set consisting of one such inequality for each process --- resembles the usually finite set of linear constraints in a typical finite-variable linear programming optimization problem. In such a problem, the existence of a solution satisfying the constraints, let alone an optimal one, is by itself a question to be considered at the outset. In fact, though, this is the kind of question the thermodynamics pioneers were grappling with, but in a far more difficult (infinite-dimensional) setting.

Today, questions of the existence of one or more functions satisfying powerful constraints, notably in the form of a very large (perhaps infinite) collection of linear integral inequalities, are most likely to be resolved with current tools of modern functional analysis, in particular the Hahn-Banach Theorem \cite{brezis2011functional,	choquet1969lectures,rudin_functional_1991,simon2011convexity}.\footnote{An account \cite{narici1997hahn} of the history of the Hahn-Banach Theorem indicates how, even in the early 20\textsuperscript{th} century, giants of mathematics were struggling to make now-standard ideas precise.}

\subsection{Motivation and Questions Considered}\label{subsec:Motivation-Questions} Although this article is about foundations of classcial thermodynamics generally, it is motivated in part by questions connected to use of the Clausius-Duhem inequality in modern continuum thermomechanics. A 1963 paper by Coleman and Noll \cite{coleman1963thermodynamics}, largely about the thermodynamics of viscous fluids, inspired a substantial body of work on the thermodynamics of a variety of materials subject to rapid deformations and to  heat transfer accompanied by sharp variations of temperature. A good example is a paper by Bowen \cite{bowen1968thermochemistry} on the thermochemistry of reacting mixtures. For a simpler paper with similar methodology see \cite{coleman-mizel1964}.

In all of this,  the Clausius-Duhem inequality was itself invoked as a
\emph{statement} of the Second Law, with a \emph{local} entropy density (entropy per mass) and a \emph{local} thermodynamic temperature within a material body taken as primitive entities. Moreover, it was generally presumed among other things that, for the material under consideration, there is a specific-entropy \emph{function} that assigns the local specific entropy to the local material state, with that state described by a list of local attributes particular to the material under consideration. 

Such presumptions, however,  might be regarded to be an unwarranted extension to nonequilibrium settings of ideas put forward by the  mid-19\textsuperscript{th} century pioneers. Indeed, despite the Gibbs remarks quoted earlier, consensus thinking even now is reflected in a response to a query put to Google Gemini in January, 2026:
\begin{quote}
\textit{Query:} In classical thermodynamics, is entropy defined only on equilibrium states? If so, why? 

\textit{Response (Edited for brevity):} In classical thermodynamics, the answer is yes: entropy is strictly defined only for systems in thermodynamic equilibrium. To calculate the entropy of a state, you must be able to imagine a reversible path from a reference state to that state. By definition, a reversible process is one that moves through a continuous sequence of equilibrium states. If a system is far from equilibrium, there is no reversible path that can ``land" on that specific chaotic state, making the classical calculation impossible.
\end{quote}

With this in mind, we consider two questions that might be posed for thermodynamic theories of particular material classes (e.g., elastic solids, viscous fluids, mixtures with chemical reactions).
\begin{enumerate}[(i)]
\item  Does the classical Second Law itself (in particular, a  version of the Kelvin-Planck statement) ensure the \emph{existence} of two functions of the \emph{local} (perhaps nonequilibrium) material state --- a specific entropy function and a thermodynamic temperature function --- such that, with respect to those local state functions, the Clausius-Duhem inequality is satisfied for all processes the material admits?

\item When such entropy and thermodynamic-temperature functions do exist, under what circumstances is \emph{uniqueness} of those functions ensured, up to inconsequential changes of scale?
\end{enumerate}

\begin{rem}[Vocabulary: \emph{State} and \emph{Condition}] Throughout this article we adopt the viewpoint of continuum mechanics, in which the idea of a material point is fundamental, and reference is freely made to the density, velocity, stress tensor, temperature, or species concentrations at a material point. We will always regard a \emph{state} as an attribute of a material point within a body,
not of the body as a whole. For the body as a whole we will refer to its \emph{condition}.

\end{rem}

\section{Preview: The Hahn-Banach Theorem and the Second Law} \label{sec:Preview}

It is the purpose of this section to state the version of the Hahn-Banach Theorem that we shall invoke and to suggest just why we might expect the theorem to ensure naturally, as a consequence of the Second Law, the existence of entropy and thermodynamic temperature state functions suited to the Clausius-Duhem inequality.

\subsection{A Version of the Hahn-Banach Theorem} The version of the Hahn-Banach Theorem of interest here has a highly geometric and very intuitive flavor. In a three-dimensional vector space such as $\R^3$, the theorem asserts that if $A$ and $B$ are any two disjoint closed convex\footnote{Recall that a convex set in a real vector space is a set such that, for any two members, the line segment joining them is also in the set. Recall also that a set is closed if it contains its boundary.} sets, with $B$ bounded, then $A$ and $B$ can be separated by a plane, with $A$ lying on one side of the plane and $B$ lying  on the other side strictly (i.e., with no member of $B$ lying in the plane). See Figure \ref{fig:HBIllustraion}.

\begin{figure}[H]
\centering
\includegraphics[width=0.35\textwidth]{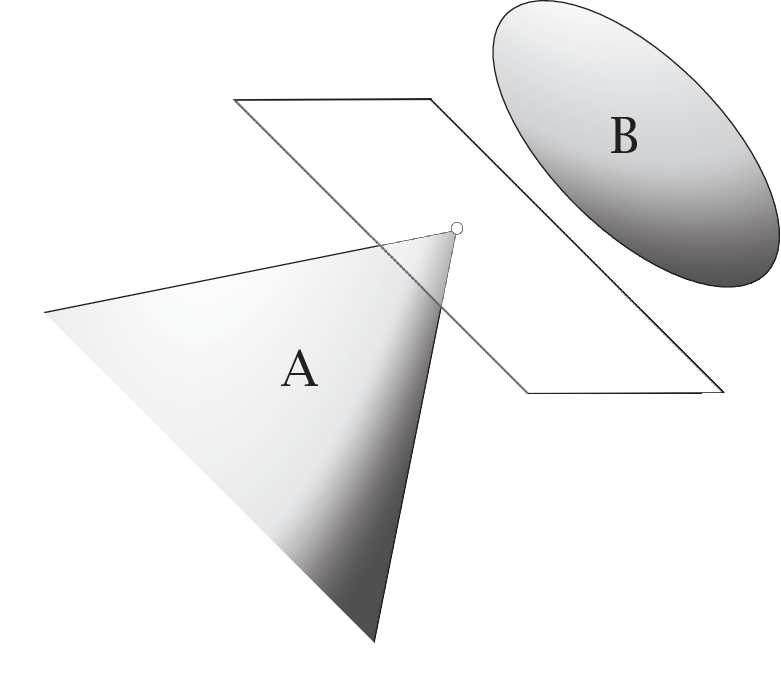}
\caption{Hahn-Banach Theorem illustration}
\label{fig:HBIllustraion}
\end{figure}

In algebraic terms, a particular plane in $\R^3$ can be identified with the set of all solutions of an equation $f(x) = \gamma$, where $f(\cdot)$ is a linear real-valued function on $\R^3$ and  $\gamma$ is a number; the pair $f(\cdot)$ and $\gamma$ then characterize the plane. In effect, then, the Hahn-Banach Theorem asserts the existence of a linear function $f(\cdot)$  and a number $\gamma$ such that 
\begin{equation}\label{eq:separation}
f(a) \leq \gamma, \,\, \textrm{for all}\,\,  a\, \textrm{in}\, A \quad \textrm{and} \quad f(b) > \gamma, \,\, \textrm{for all}\,\, b\,\,  \textrm{in}\,  B.
\end{equation}
At least in the three-dimensional case, this seems geometrically obvious and perhaps even easy to prove, until it is remembered that the argument must apply to \emph{any} two convex sets conforming to the stipulated conditions. 

It is in infinite-dimensional vector spaces that the Hahn-Banach Theorem, suitably formulated, achieves its full power: If a real vector space $V$, perhaps infinite-dimensional, has certain technical qualities in common with $\R^3$, then the assertion described in $\R^3$ remains true, provided  that the \emph{closed and bounded} requirement on $B$ is replaced by the closely related requirement that $B$ be \emph{compact} \cite{rudin_real_1987}. 

In the statement of the Hahn-Banach Theorem  just below, the stipulated mild technical conditions on the vector space $V$ are necessary to make the statement  correct. For readers unfamiliar with the terminology, it is enough to suppose that $V$ is ``ordinary."

\begin{samepage}
\begin{theorem}{\emph{\textbf{(Hahn-Banach Theorem)}}}\label{thm:HahnBanach} Let $V$ be a Hausdorff locally convex topological vector space,  and let $A$ and $B$ be non-empty disjoint closed convex subsets of $V$, with $B$ compact. There is a  continuous real-valued linear function $f(\cdot)$ on V and a number $\gamma$ such that \eqref{eq:separation} obtains.
\end{theorem}
\end{samepage}

Before turning to a discussion of connections of the Hahn-Banach Theorem to the Second Law, it will be useful to state an easy corollary. We say that $A$ is a \emph{cone} if, for every vector $a$ in $A$ and every positive number $\alpha$, the vector $\alpha a$ is also in $A$.

\begin{corollary} \label{cor:HBCone} If, in Theorem \ref{thm:HahnBanach}, $A$ is a closed convex cone, then $\gamma$ can be taken to be $0$. That is, there is a continuous real-valued linear function $f(\cdot)$ on V such that
\begin{equation}\label{eq:separation0}
f(a) \leq 0, \,\, \textrm{for all}\,\,  a\, \textrm{in}\, A \quad \textrm{and} \quad f(b) > 0, \,\, \textrm{for all}\,\, b\,\,  \textrm{in}\,  B.
\end{equation}
\end{corollary}

\subsection{The Hahn-Banach Theorem as a Natural Companion of the Second Law}  With this as background, we can begin to see a sense in which the Second Law of Thermodynamics provides natural and fertile grounds for application of the Hahn-Banach Theorem. Here, in this introduction, we can only describe in vaguely suggestive terms how this connection comes about. More detail will follow.

Suppose that, for bodies composed of a particular material, all conceivable processes  --- actual or imagined, whether Second-Law-compliant or not ---  can be somehow associated with members of a vector space $V$ of the kind described in the theorem statement. Let $P$ denote the set of vectors corresponding to all \emph{actual} processes the material admits. In this introductory section we will assume that, for each process vector in $P$, any positive multiple of it is also a member of $P$, perhaps corresponding to a process that is a scaled copy of the original process. That is, we are assuming in this section that $P$ is a cone. Furthermore, let $X$ denote the set of all non-zero vectors in $V$ that, by virtue of the Second Law, could not represent actual processes. It is a tautology, then, that $P$ and $X$ are disjoint. 

Note, however, that classical thermodynamics carries with it an implicit assumption: No hypothetical process (such as a Carnot cycle) that is not itself an actual process but which can be approximated arbitrarily closely by actual processes, can violate the Second Law. Stated differently, there is the implicit assumption that the actual processes must not only satisfy the Second Law, \emph{they must not come arbitrarily close to violating it}. In the limited context of this section, \emph{the Second Law will be understood to be the assertion that the \emph{closure of $P$} (denoted $cl(P)$) --- i.e., $P$ taken with its boundary ---  is disjoint from $X$  (and, therefore, from any compact convex subset of $X$, say $X^{\ast}$}).

Now if it can be argued that $cl(P)$ is a closed \emph{convex} cone --- i.e., that the sum of any two members of $cl(P)$ is also a member --- then the Hahn-Banach Theorem would assert that there is a linear continuous real-valued function $f(\cdot)$ on $V$  such that
 
\begin{equation}\label{eq:separation2}
f(p) \leq 0, \,\, \textrm{for all}\,\,  p\; \textrm{in}\, cl(P) \quad \textrm{and} \quad f(x) > 0 \,\, \textrm{for all}\,\, x\,\,  \textrm{in}\,  X^{\ast}.
\end{equation}

Note that the first of these inequalities asserts that t\emph{here exists a linear function $f(\cdot)$ on $V$ such that the inequality  $f(p) \leq 0$ is satisfied by every vector $p$ corresponding to an actual process.} This already has something of the flavor of the assertion, described in $\S$\ref{subsec:Legacy}, that \emph{there is an integral Clausius-Duhem inequality, linear in entropy and reciprocal temperature functions of the local material state, that is satisfied by all actual processes (and their limits) that the material under study admits}.

This impressionistic preview of course needs to be fleshed out considerably, but it offers a hint of how the Hahn-Banach Theorem can provide a connection between, on one hand, the Second Law and, on the other hand, the existence of entropy and thermodynamic temperature functions of state suited to the Clausius-Duhem inequality: The Second Law first provides the disjunction between the set of actual processes from the set of hypothetical forbidden ones.  Then, in a vector space representation,  the Hahn-Banach Theorem (taken with some mild physical suppositions)  ensures the existence of a linear inequality separating one set from the other.

	In the next section we will begin to consider how, in a highly simple cartoon setting, processes can be given vector space representation and how the Hahn-Banach Theorem already provides considerable thermodynamic and geometric  insight, derived from easily-depicted two-dimensional diagrams. The lessons taught by the example will provide useful guidance for the understanding of the more general and more mathematical sections that follow.
 
\section{An Instructive Cartoon Example} \label{sec:cartoon}

This section is  intended to be only instructive. It is highly restricted in two ways. First, the focus will be entirely on the Clausius-Duhem inequality \emph{restricted to cyclic processes}. Second, we will limit our considerations to the thermodynamics of a thoroughly unrealistic toy substance in which material points reveal themselves in only two possible states (as distinct from, say, a diffusive reacting mixture, for which there are an infinite number of possible local states, characterized in part by the varying local mixture composition). For such a toy substance, we will be able to draw motivating pictures in the familiar two-dimensional vector space $\R^2$.

\subsection{The Clausius-Duhem Inequality Restricted to Cyclic Processes} Later on, and far more generally, we will show how, from a statement of the Second Law, the Hahn-Banach Theorem ensures the existence of entropy and thermodynamic temperature functions of the local material state such that the resulting Clausius-Duhem inequality is satisfied for \emph{every} process  a particular material might admit. In this expository section our aims will be more modest, in consideration of the weaker requirement that the resulting Clausius-Duhem inequality be satisfied only by all cyclic processes. 

	By a \emph{cyclic process}, we mean in this section a process  in which the condition of the body experiencing the process is the same at the end of the process as it was at the beginning. Because, for a cyclic process, the total entropy of the body experiencing the process is the same at the process's beginning and end, the  Gibbs formulation \eqref{eq:CDInGibbs} of the Clausius-Duhem requirement takes a restricted form: For every \emph{cyclic} process,

\begin{equation}
\label{eq:CDInGibbsCyclic}
0
\quad\geq\quad
{\int\frac{dq}{T}}\ \biggr\rvert_{\;\textrm{process}},
\end{equation}
\noindent 
\emph{``dq denoting the element of heat received from external sources and $T$ denoting the temperature of the part of the system receiving it."}

For historical reasons, the Clausius-Duhem inequality requirement, restricted to cyclic processes, is usually (and confusingly) called  the \emph{Clausius inequality} requirement. Note that in \eqref{eq:CDInGibbsCyclic} there is no mention of entropy. For that reason, our limited goal in this section will be to consider how, from a statement of the Second Law, the Hahn-Banach Theorem gives rise  to a thermodynamic temperature function of state $T(\cdot)$  such that \eqref{eq:CDInGibbsCyclic} is satisfied for every cyclic process. This we call a Clausius temperature scale.\footnote{In \cite{feinberg1983thermodynamics} we invoked the Hahn-Banach theorem to examine the existence, uniqueness, and properties of a Clausius temperature scale, as consequences of the Kelvin-Planck Second Law, with that function having as its domain a general state space similar to the one in \S\ref{subsect:StateSpace}. Working independently, Miroslav Šilhavý \cite{silhavy1980measures} too invoked the Hahn-Banach theorem to derive the existence of a Clausius temperature scale, but one having as its domain a presumed empirical temperature scale. In the same empirical temperature scale setting, we considered existence and uniqueness of a Clausius temperature scale, via the Hahn-Banach Theorem, in a private communication to James Serrin \cite{feinberg-lavine_PrelimNotes}.}

\subsection{A Toy Two-State Substance} \label{subsec:ToyTwo-State} As indicated, we confine our attention in this section to the thermodynamics of a hypothetical toy substance such that, in any body composed of the substance, a material point can be in only one of two local states, labeled $1$ and $2$. \emph{We do not insist that, at a given instant, all material points within the body are in the same state, and we allow for the possibility that, as time evolves,  the various material points can switch from one state to the other. }

\subsubsection{Heating vectors} With each process experienced by a body composed of the substance under consideration we associate a \emph{heating vector}, $\scrq := [\scrq_1,\scrq_2]\  \textrm{in}\  \R^2$, having the following interpretation: For $i =1,2$, $\scrq_i$ 
is the net amount of heat (in, say, calories) received \emph{from the exterior of the body} during the course of the \emph{entire} process by material which, \emph{at the time of heat receipt}, is in state $i$. In colloquial terms, to determine $\scrq_i$ we might observe the process through lenses that show only material in state $i$ and record at the end of the process the net amount of heat --- positive, negative, or zero ---  received by \emph{visible} material from the exterior of the full body experiencing the process. (What is visible might change from moment to moment.) Note that $\scrq_1 + \scrq_2$ is the net amount of heat received by the body over the course of the entire process from the body's exterior.

The First Law of Thermodynamics requires that, for a cyclic process, the net work done by the body experiencing the process on its exterior is equal to the net heat receipt by the body from its exterior. Thus, if  $\scrq$ is the heating vector for a cyclic process experienced by a body of our two-state substance, \emph{the net amount of work done by the body is  $\scrq_1 + \scrq_2$}.

\subsubsection{Tentative statement of the Kelvin-Planck Second Law, applied to a two-state substance} In terms that will be made more precise later on, the Kelvin-Planck version of the Second Law requires that, \emph{in a cyclic process, the body experiencing the process cannot merely absorb heat from its exterior; there must be some (qualitatively different) heat emission as well}. In this way, the Kelvin-Planck Second Law denies the possibility of a perfectly efficient cyclic engine that converts all heat absorbed from the exterior entirely into work.

 Returning to our two-state substance, we denote by $\scrC$  the set in $\R^2$ consisting of all \emph{cyclic heating vectors} --- that is, the set of all heating vectors corresponding to cyclic processes that bodies composed of our toy material actually admit. If $\scrq := [\scrq_1,\scrq_2]$ is a  nonzero member of $\scrC$, then the Kelvin-Planck Second Law requires that $\scrq_1$ and $\scrq_2$ cannot both be non-negative; if one is positive, then the other must be negative. If we denote by $\R^{\,2}_{\geq 0}$ the set of all vectors of $\R^2$ having only non-negative components, then, for our toy substance, the Kelvin-Plank Second Law requires that $\scrC$ meets  $\R^{\,2}_{\geq 0}$ at most in the zero vector of $\R^2$. In fact, the Second Law carries the implicit requirement that vectors of $\scrC$ not come arbitrarily close to a fixed non-zero vector of $\R^{\,2}_{\geq 0}$. This is expressed in the form

\begin{equation}\label{eq:TentSecondLaw}
cl\,(\scrC)\, \cap \, \R^{\,2}_{\geq 0}\  \textrm{is at most}\, [0,0],
\end{equation}
where $\cap$ indicates intersection and again $cl\,(\cdot)$ indicates the \emph{closure} of a set --- i.e., the set taken with its boundary. In a subtle way, however, this falls just a little short of ensuring that no cyclic process \emph{can come arbitrarily close to having perfect efficiency}, as we shall see in the next section.

\subsubsection{A slightly stronger Second Law: the denial of an approach to perfect efficiency} \label{subsubsec:EffCounteExamp}


It is not unreasonable to suppose that $\scrC$ is closed under multiplication by any positive number, not just the positive integers, in which case $\scrC$ would be a cone. For if $\scrq$ in $\scrC$ corresponds to a particular cyclic process and $\alpha$ is a positive number, then one can imagine a scaled copy of that process, of perhaps different duration and suffered by a body of different mass, giving rise to the cyclic heating vector $\alpha\scrq$.

If, however, $\scrC$ is not a cone, then \eqref{eq:TentSecondLaw} --- our tentative statement of the Second Law ---  is not sufficiently strong as to preclude the possibility that cyclic processes encoded by $\scrC$ come arbitrarily close to achieving perfect efficiency. This will be demonstrated with an example.

\begin{example}\label{ex:EffExample}
 Let $\scrC$ be the closed set in $\mathbb{R}^2$ defined by 
\begin{equation}
\scrC := \{\scrq\  \mathrm{in}\  \mathbb{R}^2:\,\scrq_1 < 0,\, \scrq_2>0,\, \scrq_1\scrq_2 \,\leq\  -1\}. \label{eq:EffExample}
\end{equation}
It is evident from Figure \ref{fig:Hyperbola}, in which are depicted a variety of heating vectors, that $\scrC$ satisfies condition \eqref{eq:TentSecondLaw}.

\begin{figure}[H]
\centering
\includegraphics[width=0.5\textwidth]{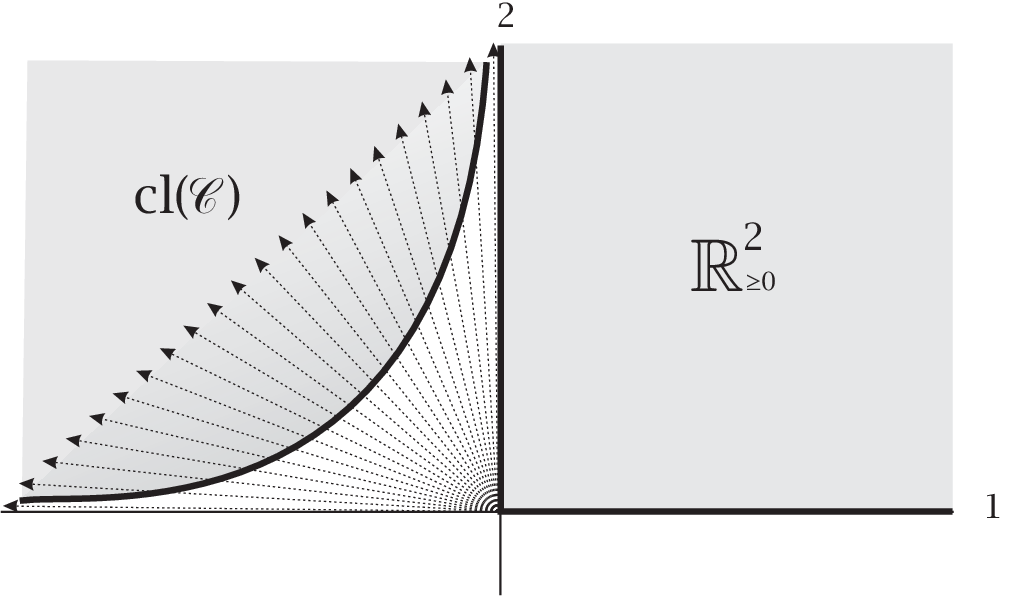}
\caption{An approach to perfect efficiency}
\label{fig:Hyperbola}
\end{figure}

Nevertheless, the processes represented in $\scrC$ come arbitrarily close to achieving perfect efficiency, which for our example is defined in the following way: If $\scrq$ is an element of \scrC, then the efficiency of the cyclic process associated with $\scrq$ is the work done by the body experiencing the process, $\scrq_1 + \scrq_2$, divided by the heat received by the body from its exterior, $\scrq_2$:
\begin{equation}\label{eq:EffDef}
\mathrm{eff}\,(\scrq) :=\  \frac{\scrq_1 + \scrq_2}{\scrq_2}.
\end{equation}
Because in our example $\scrq_1$ is always negative, no cyclic process has an efficiency of $1$.

	However, Figure \ref{fig:Hyperbola} indicates that, along the hyperbolic boundary of $\scrC$, \emph{the process efficiency comes arbitrarily close to 1 as $\scrq_1$ approaches 0}, this despite the fact that \eqref{eq:TentSecondLaw} is satisfied. In the example, the failure of \eqref{eq:TentSecondLaw} to prelude an approach to perfect efficiency resides in the fact that, although each nonzero element of  $\R^{\,2}_{\geq 0}$ is bounded away from $\scrC$, the members of $\scrC$ come arbitrarily close to \emph{aligning} with a nonzero member of $\mathbb{R}^2_{\geq 0},$, in particular the vector $[0,1]$. 
	
	This ends our example, which demonstrates that when $cl(\scrC)$ is not a cone, a Second Law of the form \eqref{eq:TentSecondLaw} can fall short of denying an approach to perfect efficiency.
\end{example}
\centerline{***}
\medskip

	From \eqref{eq:EffDef} it is evident that if $\scrq$ is a member of $\scrC$ and $\alpha \geq 1$, then $\textrm{eff}\,(\alpha\scrq) = \textrm{eff}\,(\scrq)$. This is to say that it is the \emph{direction} along which $\scrq$ points, not its magnitude, that determines its efficiency. To deny an approach to an efficiency of 1 by the set of cyclic processes, then, is to deny the possibility that members of $\scrC$ come arbitrarily close to being \emph{positive in direction} --- that is, to aligning with a nonzero member of $\mathbb{R}^2_{\geq 0}$.
	
	We can make this stricture precise in the following way. First we let $\Cone(\scrC)$ be the cone generated by $\scrC$ --- that is, the set of all vectors of $\mathbb{R}^2$ that point in the same direction as members of $\scrC$:
\begin{equation}
\Cone(\scrC)\, :=\, \{\alpha\scrq \in \mathbb{R}^2 : \alpha > 0, \scrq \in \scrC\}.
\end{equation}
Then, we require that no nonzero member of  $\mathbb{R}^2_{\geq 0}$ can be approximated arbitrarily closely by members of $\Cone(\scrC)$:

\begin{equation}\label{eq:TentSecondLaw2}
cl\,(\Cone(\scrC))\, \cap \, \R^{\,2}_{\geq 0}\,=\, \{[0,0]\}.
\end{equation}
Although $\scrC$ in Example \ref{ex:EffExample} satisfies condition \eqref{eq:TentSecondLaw}, it fails to satisfy the slightly stronger condition \eqref{eq:TentSecondLaw2}, which in turn denies the possibility of cyclic-process perfect efficiency.

	With this in mind, we will invoke condition \eqref{eq:TentSecondLaw2} as an expression of the Kelvin-Planck Second Law, applied to our cartoon two-state substance. Note that if $\scrC$ is itself a cone, there is no difference between conditions \eqref{eq:TentSecondLaw} and \eqref{eq:TentSecondLaw2}. With $\scrC$ denoting the set of cyclic heating vectors for a particular two-state substance, we hereafter employ the following abbreviation:
	
\begin{equation}
\scrChat := cl\,(\Cone(\scrC)).
\end{equation}

The Second Law takes then the form
\medskip

\noindent\emph{Second Law (Two-State Substance)}. \emph{For a two-state substance having cyclic heating vectors $\scrC \subset \mathbb{R}^2$, it must be the case that}
\begin{equation}\label{eq:TwoStateSecondLaw}
\scrChat\, \cap \, \R^{\,2}_{\geq 0}\,=\, \{[0,0]\}.
\end{equation}

The hypothetical set $\scrC$ given in \eqref{eq:EffExample} is not compliant with the Second Law as expressed by \eqref{eq:TwoStateSecondLaw} because $\scrChat$ and $\R^{\,2}_{\geq 0}$ have the vector $[0,1]$ in common.

\subsubsection{The convexity of $\scrChat$}

In general, the set of $\scrC$ of all cyclic heating vectors for our two-state substance can be expected to have an inherent structure, based on natural physical premises. For example, if $\scrq$ and $\scrq'$ in $\scrC$ correspond to two physical processes of the same duration, then $\scrq + \scrq'$ should also be a member of $\scrC$, for the two processes could be run simultaneously in distant locations on two separate bodies to produce a new process, experienced by the two bodies, this time viewed as one. As a simple corollary, for each $\scrq$ in $\scrC$ and for every positive integer $n$, we should expect $n\scrq$ to be a member of $\scrC$.

Based on similar natural premises (and in a far more general setting) it was argued in \cite{feinberg1983thermodynamics} that we should expect the closed cone $\scrChat$ to be \emph{convex}. Premises of this kind  are discussed more fully in the next section, where we will suggest in more detail why, in consideration of all processes, not just cyclic ones, a far richer set $\scrPhat$ is convex.

Meanwhile, we shall suppose hereafter  in this section that, for our two-state substance, $\scrChat$ is a closed \emph{convex} cone.

\subsubsection{How the Second Law and the Hahn-Banach Theorem immediately give rise to a Clausius temperature scale}
For the two-state substance under consideration, the Second Law, as expressed in \eqref{eq:TwoStateSecondLaw}, requires that the closed convex cone $\scrChat$ be disjoint from the closed, bounded, and convex line segment $L$ in $\R^{\,2}_{\geq 0}$ with endpoints $[0,1]$ and $[1,0]$. 
See Figure \ref{fig:ScndLawHB}.

\begin{figure}[H]
\centering
\includegraphics[width=0.5\textwidth]{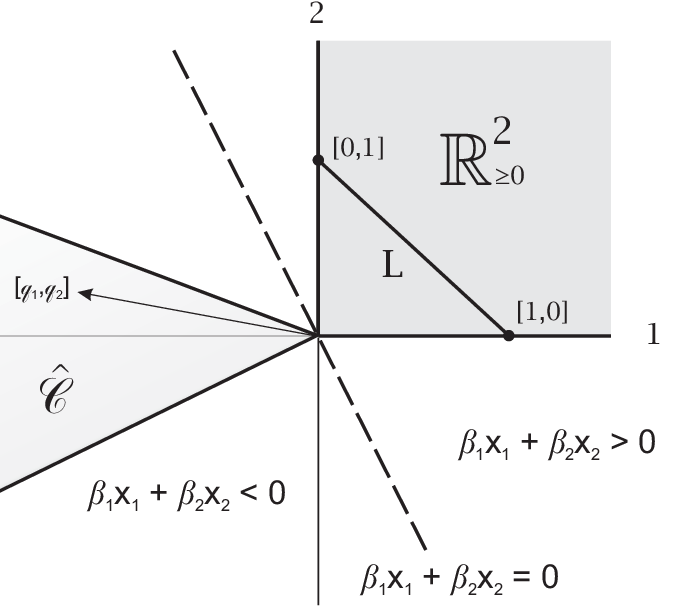}
\caption{The Second Law and the Hahn-Banach Theorem}
\label{fig:ScndLawHB}
\end{figure}

The Hahn-Banach Theorem, in the form of Corollary \ref{cor:HBCone}, then requires the existence of a linear real-valued function $\beta(\cdot)$ on $\R^2$, expressed as
\begin{equation} \label{eq:BetaDef}
\beta(x)\,:=\,\beta_1x_1 \, + \, \beta_2x_2,
\end{equation}
such that 
\begin{equation} \label{eq:ChatIneq}
\beta(\scrq)\,=\,\beta_1\scrq_1 \, + \, \beta_2\scrq_2\, \leq\,\, 0 \,\, \ \mathrm{for\, all} \,\, \scrq\,\, \mathrm{in}\,\, \scrChat
\end{equation}
and
\begin{equation}\label{eq:LIneq}
\beta(x)\,=\,\beta_1x_1 \, + \, \beta_2x_2 \, > \, 0 \,\,\  \mathrm{for \,\, all} \,\, x\,\, \mathrm{in}\,\, L.
\end{equation}
Because both $[1,0]$ and $[0,1]$ are members of $L$, it follows from \eqref{eq:LIneq} that both $\beta_1$ and $\beta_2$ are positive. Letting $T_1$ and $T_2$ be the reciprocals of $\beta_1$ and $\beta_2$, we can invoke \eqref{eq:ChatIneq} to write
\begin{equation} \label{eq:TwoStateClausiusIneq}
0 \,\, \geq \,\,\frac{\scrq_1}{T_1} \, + \, \frac{\scrq_2}{T_2},  \,\, \mathrm{for\, all} \,\, \scrq\,\, \mathrm{in}\,\, \scrC.
\end{equation}
If we regard the ``function of state"  $T(\cdot)$ that assigns the state $i$ in $\{1,2\}$ to the positive number $T_i$, \emph{then \eqref{eq:TwoStateClausiusIneq} becomes, for our two-state substance, an expression of the Clausius inequality as posited by \eqref{eq:CDInGibbsCyclic} (and the text immediately following), with $T(\cdot)$ playing the role of a Clausius temperature scale}.

	Note that the existence of such a Clausius temperature scale, satisfying the Clausius inequality for every process, follows directly and immediately from the Second Law, via the Hahn-Banach Theorem. There is no invocation of reversible cyclic processes, much less Carnot cycles, and no notion of equilibrium comes into play.
	
\subsubsection{The uniqueness of a Clausius temperature scale}\

Relative to a specified set $\scrC$ of cyclic heating vectors, it is evident that if $T(\cdot)$ satisfies the Clausius inequality \eqref{eq:CDInGibbsCyclic}, then $\alpha T(\cdot)$ will do the same for all positive $\alpha$. This is to say that a Clausius temperature scale can be unique only up to positive multiplication. For this reason, we will say that a specified $\scrC$ admits an \emph{essentially unique} Clausius temperature scale if all Clausius temperature scales are positive multiples of some fixed one.

\begin{figure}[H]
\centering
\includegraphics[width=0.4\textwidth]{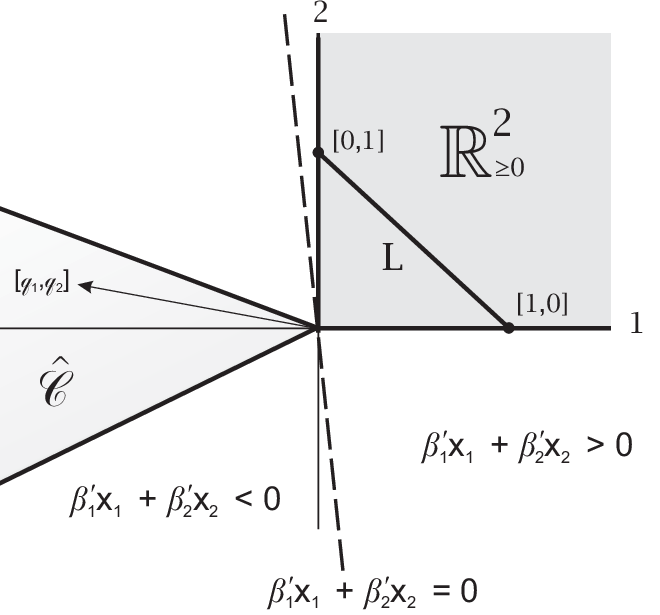}
\caption{A different linear separating function}
\label{fig:ScndLawHB2}
\end{figure}

	Taken together, Figures \ref{fig:ScndLawHB} and \ref{fig:ScndLawHB2} make clear that, for the same set of cyclic processes, there can be essentially different Clausius temperature scales. There is a variety of lines, having different orientations, that separate $\scrChat$ from $L$.
	
	In general, the more cyclic processes there are --- in particular, the broader $\scrChat\,$ becomes --- the more demanding will it be for a temperature scale to satisfy the Clausius inequality for every cyclic process. It is only when the closed convex cone  $\scrChat$ is sufficiently broad that essential uniqueness of a Clausius temperature scale can be expected. For our two-state substance, there are precisely two circumstances, as depicted in Figures \ref{fig:UniquenessLine} and \ref{fig:UniquenessHalfSpace}, in which such uniqueness will be realized: when $\scrChat$ is a full line or when $\scrChat$ is a half-space.
	
\begin{figure}[H]
\centering
\includegraphics[width=0.5\textwidth]{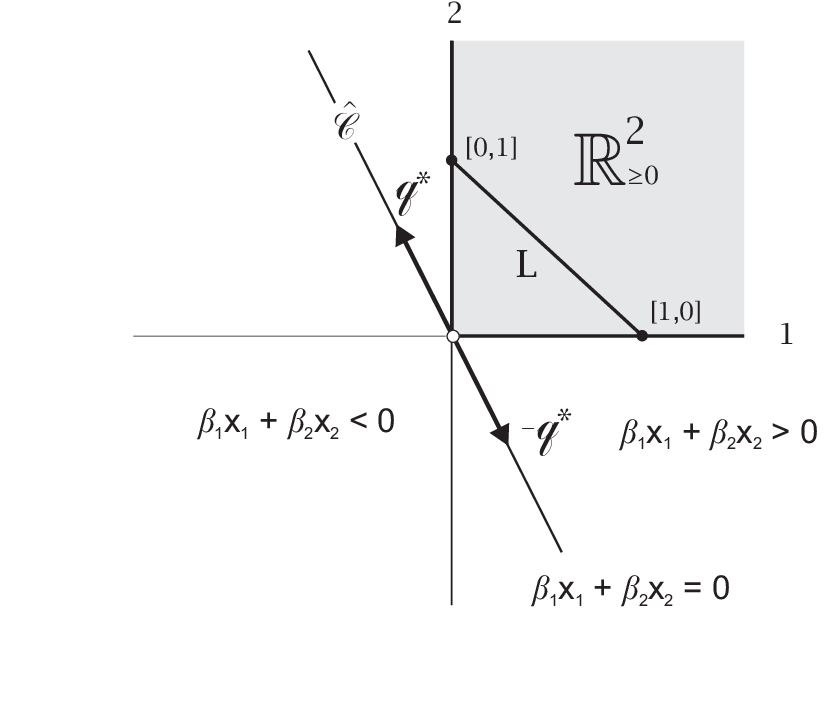}
\caption{Clausius temperature scale uniqueness: $\scrChat$ is a line}
\label{fig:UniquenessLine}
\end{figure} 

In the the first instance, depicted in Figure \ref{fig:UniquenessLine}, we imagine that, for specified positive numbers $\alpha_1$ and $\alpha_2$,  $\scrC$ is a full line: the set of all $\scrq$ in $\R^2$ that satisfy the equation $\alpha_1\scrq_1 + \alpha_2\scrq_2 = 0$. In this case $\scrChat = \scrC$, and in \eqref{eq:BetaDef} - \eqref{eq:LIneq} there is no choice (up to a positive common factor) other than $\beta_1 = \alpha_1$ and $\beta_2 = \alpha_2$, with equality holding in \eqref{eq:ChatIneq}. The resulting Clausius temperature scale, $T_1 = \frac{1}{\alpha_1}, T_2 = \frac{1}{\alpha_2}$  is essentially unique.

\begin{figure}[H]
\centering
\includegraphics[width=0.5\textwidth]{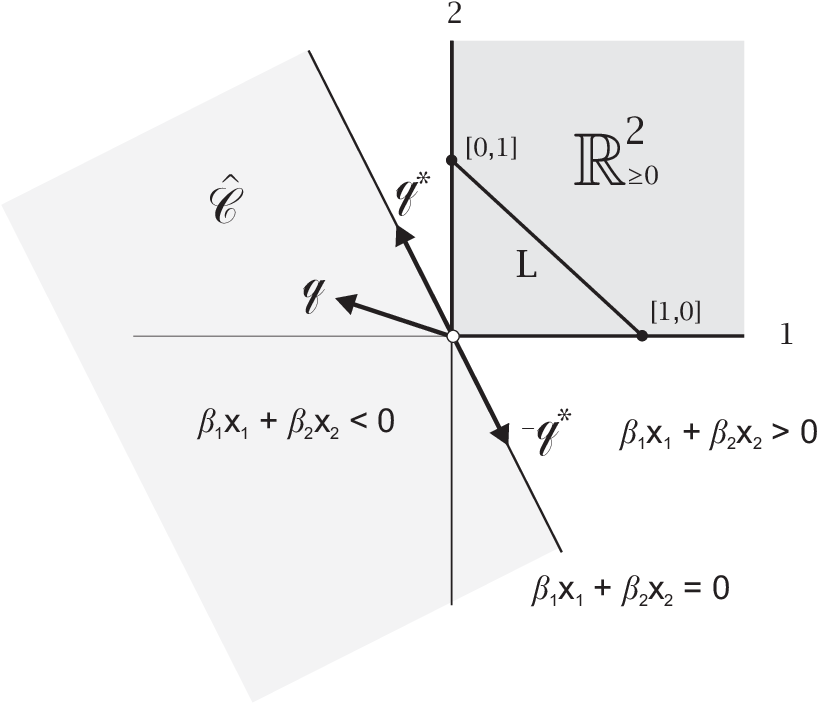}
\caption{Clausius temperature scale uniqueness: $\scrChat$ is a half-space}
\label{fig:UniquenessHalfSpace}
\end{figure} 

In the the second instance, depicted in Figure \ref{fig:UniquenessHalfSpace}, we imagine that, for specified positive numbers $\alpha_1$ and $\alpha_2$,  $\scrC$ is the set of all $\scrq$ in $\R^2$ that satisfy the inequality $\alpha_1\scrq_1 + \alpha_2\scrq_2 < 0$. In this case $\scrChat$ is a closed half-space: the set of all $\scrq$ that satisfy $\alpha_1\scrq_1 + \alpha_2\scrq_2 \leq 0$. In \eqref{eq:BetaDef} - \eqref{eq:LIneq} there is again no choice (up to a positive common factor) other than $\beta_1 = \alpha_1$ and $\beta_2 = \alpha_2$.  The resulting Clausius temperature scale is again essentially unique.

\medskip

Note that in both instances there is a \emph{reversible} element $\scrq^*$ in $\scrChat$ --- that is, a member of $\scrChat$ whose negative is also a member of $\scrChat$. If $\scrq^*$ is not itself a member of $\scrC$ (as distinct from $\scrChat$) then it is nevertheless approached arbitrarily closely by members $\scrC$, corresponding to heating vectors of actual cyclic processes, or positive multiples of them.

Note also that, at least in this cartoon two-state substance example, if there is an essentially unique Clausius temperature scale and there is at least one \emph{irreversible} heating vector in $\scrC$, then $\scrChat$ must contain \emph{all} vectors consistent with the Clausius inequality.

\subsection{Lessons of the Cartoon Example} In their almost simultaneous arguments for both the \emph{existence} and \emph{essential uniqueness} of a Clausius temperature scale, the 19\textsuperscript{th} century thermodynamics pioneers invoked the ubiquity of reversible processes, especially Carnot cycles, from the outset. For this reason, it becomes difficult to disentangle, on one hand, the role played by the Second Law itself from, on the other hand, the role played by the assumption that the various material states considered can invariably be visited by reversible processes. Existence and uniqueness are, of course, different matters, as are necessary and sufficient conditions for both.

The cartoon example considered in this section illustrates how these things can become conflated. For \emph{existence} of a Clausius temperature function of state, the Second Law by itself was \emph{sufficient}; that such a temperature scale exists derives from the Second Law alone, via the Hahn-Banach Theorem.  For that temperature scale to be \emph{essentially unique}, it was \emph{necessary} that $\scrChat$ contain reversible elements.

In the next section, we will abandon our cartoon to take up, far more generally, existence and uniqueness of two functions of the local material state, giving the local specific entropy and the local thermodynamic temperature, that together satisfy the full Clausius-Duhem inequality for all processes. Nevertheless, geometric reasoning and lessons learned in consideration of the cartoon will remain highly pertinent.

\section{Thermodynamic Theories} 

In common parlance, there are thermodynamic theories of gases, of reacting mixtures, of elastic materials, and so on.

In the preceding section we studied the thermodynamics of a cartoon substance, in which material points reveal themselves in only one of two states, labeled $1$ and $2$. For that substance, then,  the ``state space" was simply a set $\{1,2\}$ consisting of just two members. Moreover, we associated a cyclic heating vector in $\R^2$ with each of the cyclic processes that bodies  composed of the substance were deemed to admit.  In this way, the presumed nature of the cartoon substance was carried by just two sets: the state space $\Sigma := \{1,2\}$  and the set $\scrC$ of heating vectors associated with the various cyclic processes. The prescription of both $\Sigma$ and $\scrC$ in effect amounted to a specification of a ``thermodynamic theory" of the two-state substance --- in particular, a specification of just those the details essential to investigation of whether or not presumptions about the substance are consistent with one or another statements of the Second Law. 

In this section we will erect a far more general notion of a thermodynamic theory for a particular substance, again involving specification of a state space $\Sigma$ along with a set $\scrP$ of vectors (in a vector space rather different from $\R^2$), each such vector associated with a putative  process (not necessarily cyclic) that bodies composed of the substance are presumed to admit. 

This will permit us to construct something like a meta-thermodynamics wherein it becomes possible to state theorems (all involving the Hahn-Banach Theorem in their proofs) that provide universal properties of Kelvin-Planck theories --- that is, of thermodynamic theories that comply with a version of the Kelvin-Planck Second Law, in a form similar to \eqref{eq:TwoStateSecondLaw}.

\subsection{The Set of States} \label{subsect:StateSpace}To say that, for a particular substance, the specific entropy (entropy per mass) is a ``function of state" is to imply that there is a set of ``states," specified for the substance under study, that serves as the domain of the specific entropy function. For a pure substance such as water, a local state might be specified by a pair $[p,v]$ in $\R^2$, where $p$ is the local pressure and $v$ is the local specific volume (the reciprocal of the density). For a diffusive reacting mixture in which $n$ chemical species are present a local state might be specified by a vector $[c_1,c_2,\dots,c_n,\theta]$ in $\R^{n+1}$, where $c_i$ is the local molar concentration of the $i^{th}$ species, and $\theta$ is the local temperature on some empirical temperature scale. Similarly, for an elastic material, taken with some reference configuration, a local state might be specified by a vector in $\R^{10}$, the entries of which are the nine components of the local deformation gradient (relative to some basis) and the local specific internal energy.

In any case, a thermodynamic theory presumably has an intended range of applicability. In the case of the reactive mixture, the range of applicability will almost certainly not embrace molar concentrations so large that they could appear only in black holes or temperatures so high that they would be realized only on the sun. 

\emph{Hereafter, when we refer to the state space $\Sigma$ of a thermodynamic theory, it will be presumed that $\Sigma$ is a closed and bounded (and therefore compact) subset of $\mathbb{R}^N$, for some finite value of $N$, and that processes considered within the theory never realize local states outside of\, $\Sigma$.} An open set in $\Sigma$ is regarded to be one that is the intersection of $\Sigma$ with an open set of $\mathbb{R}^N$.

\subsection{The Set of Processes}\label{subsec:SetOfProcesses} 

For our cartoon substance, for which a local material point could exist in only one of two states (i.e., $\Sigma = \{1,2\}$), we identified each process with a \emph{heating vector} in  $\R^2$; that was all we needed, largely because there were just two states and our concern was only with cyclic processes, in which the condition of a body at the end of the process is the same as its condition at the beginning. Hereafter, our concerns will be far more general, in consideration of processes that might not be cyclic and of substances for which the state space might have an infinite number of elements.  With a process, we will want to describe \emph{the change of condition} of the body experiencing the process and also what we shall call the \emph{heating measure} for the process, which is analogous to the heating vector for our two-state cartoon substance. Because substances considered (such as a reacting mixture) might have an infinite number of states (each of the kind $[c_1,c_2,\dots,c_n,\theta]$, corresponding to varying temperatures and compositions), the mathematics required must be sufficiently broad as to accommodate these more general needs. 

\subsubsection{Informal mathematical preliminaries}\label{subsubsec:MathPrelims}

Consider a body of mass $M$ composed of a substance, perhaps a diffusive reacting mixture, with state space $\Sigma$. At a particular instant, we might want to describe how the states in $\Sigma$ are distributed across the body's mass. Especially when the body is in a very diffuse condition, it will be helpful to have available the idea of a \emph{positive measure on $\Sigma$}, a function that assigns a non-negative number to every open set in $\Sigma$.\footnote{As a technical matter, all measures on $\Sigma$ discussed in this article are understood to be regular Borel measures. The first two chapters of \cite{rudin_real_1987} provide a succinct and more than ample introduction to aspects of measures and integration required for this article.} At the chosen instant the condition of the body can be described by a positive measure $\scrm(\cdot)$ with the following interpretation: If $\Omega$ is an open set in $\Sigma$, then $\scrm(\Omega)$ is the mass of material having local states within $\Omega$. Note that $\scrm(\Sigma) = M$. 

	Associated with every positive measure $\scrm(\cdot)$ on $\Sigma$ and every continuous real-valued function $f(\cdot)$ on $\Sigma$ there is a real number $\int_{\Sigma}fd\scrm$, the integral of $f$ with respect to \scrm. By way of interpretation, consider a body at a fixed instant composed of a reacting mixture, with state space $\Sigma$ as indicated in $\S$\ref{subsect:StateSpace}. Suppose also that a positive measure  $\scrm(\cdot)$ at that instant describes, for the body, the distribution of mass among the states. If $v(\cdot)$ is a continuous constitutive function for the mixture that assigns, to each state $\sigma$ in $\Sigma$, the value $v(\sigma)$ of the specific volume (volume per mass) of material in state $\sigma$ (the reciprocal of the density), then $\int_{\Sigma}vd\scrm$ is the instantaneous volume of the body.
	
	Before proceeding further we note that, for each state $\sigma$ in $\Sigma$, there is a positive \emph{Dirac measure}, $\delta_{\sigma}(\cdot)$, defined by the property that, for every open set $\Omega$ in $\Sigma$, $\delta_{\sigma}(\Omega)$ is $1$ if $\sigma$ is a member of $\Omega$ and is $0$ otherwise. For every continuous real-valued function $\phi(\cdot)$ on $\Sigma$ and every $\sigma$ in $\Sigma$ it is a property of the Dirac measure that $\int_{\Sigma}\phi\, d\, \delta_{\sigma} = \phi(\sigma)$.

In consideration of heat transfer to a body composed of a substance with state space $\Sigma$, we will also have use of \emph{signed measures on $\Sigma$}, which are functions that assign a real number (perhaps negative) to every open set in $\Sigma$. Just as all real-valued functions on the same domain can be added (or subtracted) and multiplied by real numbers, there are the same natural rules for adding two signed measures and for multiplying a signed measure by a real number.

In this way, the set of all signed measures on $\Sigma$ becomes a real vector space, denoted $\scrM(\Sigma)$. The subset of $\MSigma$ consisting of all positive measures is denoted $\MSigmaPl$,\footnote{Note that $\MSigmaPl$ contains the \emph{zero measure} of $\MSigma$ --- that is, the measure that assigns $0$ to every open set of $\Sigma$.}  while $\MSigmaPlOne$ denotes the subset of all positive measures $\nu(\cdot)$ such that $\nu(\Sigma) = 1$. By $\MSigmaZ$ we mean the linear subspace of $\MSigma$ consisting of all signed measures $\mu$ such that $\mu(\Sigma) = 0$.\footnote{We will suppose that, like $\RN$, the vector space $\MSigma$ carries a notion of open and closed sets, this time given by what is called the weak-star topology \cite{rudin_functional_1991}. In ths case, $\MSigma$ is a vector space having all the attributes required in our statement of the Hahn-Banach Theorem, and $\MSigmaPlOne$ is compact \cite{choquet1969lectures}.}

As with positive measures on $\Sigma$, there is associated with each signed measure $\mu(\cdot)$ on $\Sigma$ and each continuous function $f(\cdot)$ on $\Sigma$ a real number $\int_{\Sigma}fd\mu$, the integral of $f$ with respect to $\mu$. Note that, for a specified continuous function $f(\cdot)$ on $\Sigma$, the function on $\MSigma$ that takes each $\mu$ in $\Sigma$ into $\int_{\Sigma}fd\mu$ is a real-valued continuous linear function on the vector space $\MSigma$. In fact, this statement has something like a converse that will be important: For every continuous linear real-valued $G(\cdot)$ on $\MSigma$ there is a continuous real-valued function $g(\cdot)$ on $\Sigma$ such that $G(\mu) = \int_{\Sigma}g\,d\mu$ for all $\mu$ in $\MSigma$.

\subsubsection{The change of condition for a process} Consider a process experienced by a body composed of a substance with state space $\Sigma$. By the \emph{condition} of the body at some fixed instant during the course of the process we mean a positive measure $\scrm(\cdot)$ on $\Sigma$ with the interpretation given earlier: For each open set $\Omega$ in $\Sigma$, $\scrm(\Omega)$ is the mass, at that instant, of all  material in the body having states within $\Omega$. 

If $\scrm_f(\cdot)$ is the final condition of the body at the end of the process and  $\scrm_i(\cdot)$ is its initial condition at the beginning of the process, then 
\begin{equation}
\deltam(\cdot) := \scrm_f(\cdot) - \scrm_i(\cdot)
\end{equation}
is a signed measure on $\Sigma$, called the \emph{change of condition} for the process. Because $\scrm_f(\Sigma)$ and $\scrm_i(\Sigma)$ are both the mass of the body experiencing the process, conservation of mass requires that $\deltam(\Sigma) = 0$. Thus, for every process the change of condition is not only a member of the vector space $\MSigma$, it is a member of the linear subspace $\MSigmaZ$.

\subsubsection{The heating measure for a process} In consideration of our two-state cartoon substance in Section \ref{sec:cartoon}, we associated with each cyclic process a \emph{heating vector} in $\R^2$, designed to capture for that process the nature of net heat receipt by material in each of the two states during the process's full duration. Here we need something more general, for the state space will almost always have an infinite number of elements. Moreover, we are now considering all processes, not just the cyclic ones,

Consider a process experienced by a body composed of a substance with state space $\Sigma$. With the process we associate a  \emph{heating measure} $\scrq(\cdot)$ in $\MSigma$ with the following interpretation: For each open set $\Omega$ in $\Sigma$, $\scrq(\Omega)$ is the net amount of heat absorbed from the body's exterior during the entire process by material, which at the time of heat transfer, was in local states contained within $\Omega$. This is the analog of the heating vector in  Section \ref{sec:cartoon}, with an interpretation similar to the one given there.

\subsubsection{The set $\scrP$ of processes for a thermodynamic theory} For a thermodynamic theory of a substance with state space $\Sigma$, we have associated with each process a change of condition $\deltam$ in $\MSigmaZ$ and a heating measure $\scrq$ in $\MSigma$. For our purposes going forward, that is all the information that will be needed for specification of a process. With this in mind, we will hereafter regard a process to be stipulated by the corresponding pair $(\deltam, \scrq)$ and refer to ``the process $\scrp := (\deltam, \scrq)$."

	Let $\VSigma$ be the vector space of all pairs $(\scrv,\scrw)$ with $\scrv$ in $\MSigmaZ$ and $\scrw$ in $\MSigma$; addition and multiplication by a real number are defined in the obvious way.\footnote{As a technical matter, we take $\MSigmaZ$ to have the topology it inherits as a subset of $\MSigma$ and $\VSigma$ to have the resulting product topology.} Thus, we can regard each process \scrp\ to be a  member of $\VSigma$. With a thermodynamic theory under consideration we denote by $\scrP$ the subset of $\VSigma$ consisting of all processes.

\subsubsection{The natural structure of the set of processes}	
	From consideration of natural physical premises, there is a certain amount of structure that we might expect the set of (not necessarily cyclic) processes to possess. For example: 
\begin{enumerate}[(\roman*)]
\item \label{ite:item1}  If $\scrp_1 = (\deltam_1, \scrq_1)$ and $\scrp_2 = (\deltam_2, \scrq_2)$ are members of $\scrP$ corresponding to two processes of the same duration, then
\begin{equation}\nonumber
 \scrp_1 + \scrp_2  = (\deltam_1 + \deltam_2, \scrq_1 + \scrq_2)
\end{equation}
should also be a member of $\scrP$, for the two processes can be run simultaneously on two different bodies in distant locations, giving rise to a new process operating on the two bodies, now viewed as one.  In particular, if $\scrp$ is a member of  $\scrP$, then so is  $n\scrp$  for any positive integer $n$, corresponding to a process of the same duration as that of the one giving rise to $\scrp$.

\item If $\scrp$ in $\scrP$ corresponds to a process of duration $\tau$, then it also corresponds to a process of duration $\frac{\tau}{m}$, where $m$ is any positive integer. For the original process can be regarded to be the result of a sequence of $m$ sub-processes, each of duration $\frac{\tau}{m}$ and each with its own process vector in $\scrP$. These sub-processes could instead be run simultaneously instead of sequentially,  to form a new process of duration $\frac{\tau}{m}$,  with the result, by (i), giving rise to the original process vector $\scrp$. 
\end{enumerate}

From these, it follows that \emph{if $\scrp_1$ and $\scrp_2$  are members of $\scrP$ corresponding to processes having durations with a rational ratio, then $\scrp_1 + \scrp_2$ is a member of $\scrP$}. For suppose that $\scrp_1$ and $\scrp_2$  correspond to processes of duration $\tau_1$ and $\tau_2$ with $\frac{\tau_2}{\tau_1} = \frac{m_2}{m_1}$, where ${m_1}$ and $m_2$ are positive integers. From (ii) it follows that $\scrp_1$ and $\scrp_2$  can be regarded to derive from processes having identical durations $\frac{\tau_2}{m_2} = \frac{\tau_1}{m_1}$. From (i) it follows that $\scrp_1 + \scrp_2$ is a member of $\scrP$.

\medskip
	To show that a cone in a vector space is convex --- in particular $\Cone(\scrP)$ in $\VSigma$ --- it is enough to show that the sum of any two vectors in the cone is also in the cone. The argument above falls short of that, for it considered only the sum of two  vectors in $\scrP$ (as distinct from $\Cone(\scrP)$), corresponding to processes with durations having a rational ratio.

Nevertheless in \cite{feinberglavine2024entropy1}, premises (i) and (ii), along with an additional one about the relationship between processes and their durations, resulted in an argument  to the effect that, for natural processes, we should expect the set $\scrPhat$ defined by
\begin{equation}
\scrPhat := cl(\Cone(\scrP))
\end{equation}
to be convex.

\subsection{Specification of a Thermodynamic Theory} \label{subsec:ThermoTheory}
For the purposes of this article we will hereafter regard a \emph{thermodynamic theory} be specified by a pair \theory, where the the \emph{state space} $\Sigma$ is a closed and bounded subset of $\R^N$ and the \emph{process set} $\scrP$ is a subset of $\VSigma$ such that $\scrPhat$\  is convex. 

\section{Kelvin-Planck Theories}
We repeat here the verbal statment of the Kelvin-Plank Second Law given in Section \ref{sec:cartoon}: \emph{In a cyclic process, the body experiencing the process cannot merely absorb heat from its exterior; there must be some qualitatively different heat emission as well}. Again, the intent is a stricture on the efficiency of cyclic processes: \emph{It is not possible for a body experiencing a cyclic process to merely absorb heat from its exterior and convert all of that heat into work, nor is it possible for processes to come arbitrarily close to a cyclic one achieving perfect efficiency.}  Our aim in this section is indicate what we mean by a Kelvin-Plank thermodynamic theory  --- that is, a thermodynamic theory (in sense of  preceding section) which reflects, in precise way, what we take to be the intent of the Kelvin-Planck Second Law.

By a \emph{cyclic process} in a thermodynamic theory \theory\, we mean a process for which the condition of the body experiencing the process is the same at both the beginning and end of the process. That is, a process $\scrp = \process$ in $\scrP$ is  \emph{cyclic} if $\deltam = 0$.

	Consider a cyclic process $\scrp = (0,\scrq)$ in $\scrP$. Suppose that the heating measure $\scrq$ takes a positive value on some open set in $\Sigma$, implying that during the course of the process represented by $\scrp$ the body suffering the process absorbs heat from its exterior. In that case, we take the  Kelvin-Planck Second Law to require that there be a different open set for which $\scrq$ takes a negative value, implying heat emission. This is to say that no member of $\scrP$ should be of the form $(0,\scrq)$, where $\scrq$ is a non-zero member of the set of positive measures, $\MSigmaPl$.  
	
	If we denote by $(0,\MSigmaPl)$ the subset of all members of $\VSigma$ of the form $(0,\nu)$, with $\nu$ in $\MSigmaPl$, then this requirement of the Kelvin-Planck Second Law can be expressed as a constraint on the set $\scrP$ of processes in the form	
\begin{equation}\
\scrP \cap \, (0,\MSigmaPl)\  \textrm{is at most}\, (0,0).
\end{equation}
A stronger requirement, forbidding even an approach by processes to a violation of this constraint, could be expressed as
\begin{equation}\label{eq:TentKPSecondLaw}
cl\,(\scrP) \cap \, (0,\MSigmaPl)\  \textrm{is at most}\, (0,0).
\end{equation}

Here again, though, the discussion of the cartoon two-state substance is instructive, especially in \S\,\ref{subsubsec:EffCounteExamp}. Counterexamples in \cite{feinberglavine2024entropy1} indicate that even the strengthened \eqref{eq:TentKPSecondLaw} falls short of denying an approach of processes to a cyclic one of perfect efficiency. The problem, as in \S\,\ref{subsubsec:EffCounteExamp}, is that members of $\scrP$ can come close ``in direction" to members of $(0,\MSigmaPl)$.  To preclude an approach to cyclic-process perfect efficiency, the ultimate intent of the Kelvin-Planck Second Law, it is necessary to strengthen \eqref{eq:TentKPSecondLaw} still further:

\begin{equation}
cl\,(\Cone(\scrP)) \cap \, (0,\MSigmaPl)\ =\, \{(0,0)\}.
\end{equation}

	By a \emph{Kelvin-Planck theory} we hereafter mean a thermodynamic theory \theory \, such that
\begin{equation} \label{eq:KPFinal}
\scrPhat \cap \, (0,\MSigmaPl)\ =\, \{(0,0)\},
\end{equation}
where $\scrPhat$ is again the closed convex cone defined by	
\begin{equation}
\scrPhat := cl(\Cone(\scrP)).
\end{equation}

\section{Hahn-Banach Existence of Temperature and Entropy Functions of State} \label{sec:HBExistence}

In this section we will show that for any thermodynamic theory that complies with the Kelvin-Planck Second Law, as expressed by \eqref{eq:KPFinal}, the Hahn-Banach Theorem (in the form of Corollary \ref{cor:HBCone}) \emph{immediately and simultaneously} ensures the existence of two continuous functions of the local state --- one giving the local specific entropy and the other giving the local thermodynamic temperature --- such that the Clausius-Duhem inequality is satisfied for all processes. In fact, the converse is true as well: the existence of such functions implies a Second Law of the form \eqref{eq:KPFinal}. 

\begin{theorem}[Existence of Entropy and Temperature] \label{thm:CDPairExistence} For a thermodynamic theory \theory, the following are equivalent:
\begin{enumerate}[(i)]
\item $\theory\,$  is a Kelvin-Planck theory.
\item There are continuous real-valued functions on $\Sigma$,\, $\eta(\cdot)$ and  $T(\cdot)$, with $T(\cdot)$ taking only positive values, such that
\begin{equation}\label{eq:CDIneqFormal}
\int_{\,\Sigma}\eta\,d\deltam\,\,\geq\,\,\int_{\,\Sigma}\frac{d\scrq}{T}\,\,\,\, \mathrm{for\,\, all}\,\,\,(\deltam,\scrq)\, \,\,\mathrm{in}\,\, \scrP.
\end{equation}
\end{enumerate}
\end{theorem}

\begin{proof} To prove that (i) implies (ii), we suppose that \theory\ is a Kelvin-Planck theory, in which case \eqref{eq:KPFinal} holds. Then the closed convex cone $\scrPhat$ in $\VSigma$ is disjoint from $(0,\MSigmaPlOne)$, the compact convex subset of $(0,\MSigmaPl)$ consisting of all of its members $(0,\nu)$ such that $\nu(\Sigma) = 1$. (The set $(0,\MSigmaPlOne)$ will play the role of the line segment $L$ in the cartoon example.)

From Corollary \ref{cor:HBCone} it follows immediately that there is a continuous linear function $f(\cdot)$ on $\VSigma$ such that

\begin{equation}\label{eq:HB<}
f((\scrv,\scrw)) \leq 0 \,\,\, \textrm{for all}\,\, (\scrv,\scrw)\,\, \textrm{in}\,\, \scrPhat
\end{equation}
and
\begin{equation} \label{eq:HB>}
 f((0,\scrw)) > 0 \,\,\, \textrm{for all}\,\, \scrw\,\,  \textrm{in}\,  \MSigmaPlOne.
\end{equation}

Now every continuous linear real-valued function on $\VSigma$ has a representation in terms of integrals. (See \cite{rudin_functional_1991}, \S3.14).) In particular, for $f(\cdot)$,  there exist continuous real-valued functions $- \eta(\cdot)$ and  $\beta(\cdot)$ on $\Sigma$ such that
\begin{equation}\label{eq:IntRep}
f((\scrv,\scrw)) = - \int_{\Sigma}\eta\,d\scrv + \int_{\Sigma}\beta\,d\scrw \,\,\,\, \mathrm{for\; all}\,\,\, (\scrv,\scrw)\,\, \,\mathrm{in}\,\,\,\VSigma.
\end{equation}

Recall from \S\,\ref{subsubsec:MathPrelims}  that, for each state $\sigma$ in $\Sigma$, there is in $\MSigmaPlOne$ a Dirac measure, $\delta_{\sigma}(\cdot)$ and that for every continuous real-valued function $\phi(\cdot)$ on $\Sigma$ and every $\sigma$ in $\Sigma$ it is a property of the Dirac measure that $\int_{\Sigma}\phi\, d\, \delta_{\sigma} = \phi(\sigma)$.

	With the representation of $f(\cdot)$ taken as in \eqref{eq:IntRep}, it follows from \eqref{eq:HB>} that 
\begin{equation}
f((0,\delta_{\sigma}) )= \int_{\Sigma}\beta\,d\delta_{\sigma} = \beta(\sigma) > 0 \,\,\,\, \mathrm{for\,\,all} \,\,\,\,\sigma\,\, \mathrm{in}\, \Sigma.
\end{equation}
That is, $\beta(\cdot)$ takes strictly positive values. Hereafter we let $T(\cdot)$ be the reciprocal of $\beta(\cdot)$. Since $\scrP$ is a subset of $\scrPhat$,  \eqref{eq:HB<} and \eqref{eq:IntRep}  ensure that \eqref{eq:CDIneqFormal} in part (ii) of the theorem statement is satisfied.

Proof that (ii) implies (i) is not difficult.
\end{proof}

	We will of course interpret $\eta(\cdot)$ and $T(\cdot)$ as functions of the local state that give values of the local specific entropy (entropy per mass) and the local thermodynamic temperature. 	Note that if $\process$ is a member of $\scrP$ corresponding to a physical process for which $\deltam = \scrm_f - \scrm_i$, where $\scrm_f$ and  $\scrm_i$ are the conditions of the body experiencing the process  at its end and its initiation, then the first term on the left of \eqref{eq:CDIneqFormal} is just the difference between the final and initial total entropies of the body, and \eqref{eq:CDIneqFormal}  embodies the Clausius-Duhem inequality as expressed in  \eqref{eq:CDInGibbs}.

\medskip
	
	Hereafter, by a \emph{Clausius-Duhem pair} for a Kelvin-Planck theory \theory\ we mean a pair $(\eta(\cdot),T(\cdot)$) of continuous functions on $\Sigma$, with $T(\cdot)$ taking strictly positive values, such that \eqref{eq:CDIneqFormal} is satisfied. A function $T(\cdot)$ is  a \emph{Clausius-Duhem temperature scale} for the theory if there exists $\eta(\cdot)$  such that $(\eta(\cdot),T(\cdot))$ is a Clausius-Duhem pair. In that case, $\eta(\cdot)$ is a \emph{specific-entropy function} for the theory (corresponding to the Clausius-Duhem temperature scale $T(\cdot))$.

\medskip
Before this section is closed, there are two important things to note: 

\begin{rem} \label{rem:WeDoNotAssert}
We do not assert that  \eqref{eq:KPFinal} is the ``correct" embodiment of the Kelvin-Planck Second Law. We assert only that, for any thermodynamic theory in the sense of $\S$\ref{subsec:ThermoTheory}, equation \eqref{eq:KPFinal} and the existence of a Clausius-Duhem entropy-temperature pair are \emph{equivalent}.
\end{rem}

\begin{rem}[\emph{For existence of entropy and thermodynamic temperature, equilibrium is irrelevant}] \label{rem:ExistenceAndEquilibrium}
In every Kelvin-Planck theory the  simultaneous \emph{existence} of specific entropy and thermodynamic temperature functions of state is an almost \emph{immediate} consequence of the Hahn-Banach Theorem. There is no invocation of conceptual machinery such as reversible processes, much less of Carnot cycles. The words \emph{equilibrium} or \emph{quasi-static} never appear, and there is no reason to suppose that elements of $\Sigma$ are restricted to equilibrium states.
\end{rem} 

\medskip
The essential \emph{uniqueness} of these functions is a different matter, as was foretold by our cartoon example and as we shall see in Sections  \ref{sec:TempUniqueness} and \ref{sec:EntUniqueness}.

\section{Digression: Hotness and Carnot Elements}

In the next section we will state a theorem that gives conditions under which, for a Kelvin-Planck theory \theory, there is an essentially \emph{unique} Clausius-Duhem pair. In particular, the Hahn-Banach Theorem will indicate that, for such uniqueness, $\scrPhat$ \emph{must} contain a large supply of (reversible, cyclic) \emph{Carnot elements}. This section is meant to provide a few preliminaries.

\subsection{Reversible and Cyclic Elements of \scrPhat} \label{subsec:Reversible&Cyc} By a \emph{reversible element} of a thermodynamic theory $\theory$ we mean a member of $\scrPhat$ such that its negative is also a member of $\scrPhat$. By a \emph{cyclic element} of the theory we mean a member of $\scrPhat$ of the form $(0,\scrq)$, with $\scrq$ some member of $\MSigma$.
Note that these need not be  processes in $\scrP$; they need only be approximated arbitrarily closely by positive multiples of processes in  $\scrP$.

\subsection{Hotness} \label{subsec:Hotness} For the thermodynamic theory \theory, we say that two states $\sigma$ and $\sigma'$ in $\Sigma$ are \emph{of the same hotness} (denoted $\sigma' \sim \sigma$) if $\scrPhat$ contains a reversible cyclic element of the form $(0, \delta_{\sigma'} - \delta_{\sigma})$, where $\delta_{\sigma'}$ and $\delta_{\sigma}$ are, as in \S\,\ref{subsubsec:MathPrelims}, the Dirac measures at $\sigma'$ and $\sigma$. 

Were $(0, \delta_{\sigma'} - \delta_{\sigma})$ an actual process in \scrP, rather than a limiting element of \scrPhat, it could be viewed as having derived from a cyclic process performed on a body in which heat is absorbed from the body's exterior only by material in state $\sigma'$, passed through the body, and then emitted from the body in equal amount only by material while in state $\sigma$. Note that there is no work. For the reverse cyclic process, heat flow is in the opposite direction.

For a thermodynamic theory \theory, ``$\sim$" is an equivalence relation in $\Sigma$; in particular, if $\sigma \sim \sigma'$ and $\sigma' \sim \sigma^{''}$ then $\sigma \sim \sigma^{''}$. In this way, $\Sigma$ becomes partitioned into \emph{hotness levels}: A hotness level is a subset $h$ of $\Sigma$ such that if $\sigma$ is a member of $h$, then $h$ contains those and only those states in $\Sigma$ that are of the same hotness as $\sigma$.

	Note that, for a thermodynamic theory \theory, the hotness relation ``$\sim$" and the partition of states in $\Sigma$ into hotness levels \emph{derives entirely from the processes the theory contains.} There is no reference to temperature. 
	
\subsection{Hotness and Temperature} For a Kelvin-Planck theory \theory, Theorem \ref{thm:CDPairExistence} ensures that there is at least one Clausius-Duhem temperature scale. But to call such a function a \emph{temperature scale} requires, at the very least, that its values indicate when two states in $\Sigma$ are of the same hotness, suitably defined. The following theorem connects hotness in the sense of $\S$\ref{subsec:Hotness} with properties of the set of Clausius-Duhem temperature scales for the theory.

\begin{samepage}
\begin{theorem}
\label{thm:EqualHotness}
 Let $\sigma \in \Sigma$ and $\sigma' \in \Sigma$ be two states of the Kelvin-Planck theory \theory. The following are equivalent:
\begin{enumerate}[(i)]
\item $\sigma$ and $\sigma'$ are of the same hotness.
\item $T(\sigma) = T(\sigma')$ on each Clausius-Duhem temperature scale for \theory.
\end{enumerate}
\end{theorem}
\end{samepage}

Theorem \ref{thm:EqualHotness} is proved in \cite{feinberglavine2024entropy2} and, in a slightly different context, in \cite{feinberg1983thermodynamics}. Proof of the implication $(i) \Rightarrow (ii)$ is elementary. The implication $(ii) \Rightarrow (i)$ requires the Hahn-Banach Theorem and is far deeper: It asserts that \emph{if no Clausius-Duhem temperature scale distinguishes between states $\sigma'$ and $\sigma$, then the processes of the theory must be such that $\scrPhat$ contains a reversible cyclic element of the form} $(0, \delta_{\sigma'} - \delta_{\sigma})$.

	The introduction of hotness and, in particular, hotness levels will be sufficient for our definition of \emph{Carnot elements} in a thermodynamic theory. These in turn will play a crucial role in Sections \ref{sec:TempUniqueness} and \ref{sec:EntUniqueness}, where we consider conditions a Kelvin-Planck theory must satisfy in order for a Clausius-Duhem entropy-temperature pair to be unique. Here, we are omitting discussion of a \emph{hotter than} relation  for a Kelvin-Planck theory (defined solely in terms of processes the theory contains) and that relation's reflection in the theory's Clausius-Duhem temperature scales \cite{feinberg1983thermodynamics,feinberg1986foundations,feinberglavine2024entropy2}.

\subsection{Carnot Elements of \scrPhat} In the classical thermodynamics literature, a Carnot cycle is a cyclic reversible process in which the body experiencing the process receives heat only at one hotness and emits heat only at a different hotness. Carnot cycles are usually not deemed to be exactly realizable, but, mainly for the purpose of reversibility, they are regarded to be approximated increasingly well by actual processes operating more and more slowly.

To capture the Carnot-cycle idea in the framework erected so far, a small amount of new vocabulary will be needed. Suppose that $\scrm(\cdot)$ is a positive measure on $\Sigma$, and suppose that, for the purpose of illustration, $\scrm(\cdot)$ describes the instantaneous condition of a body. That is, for an open set $\Omega$ in $\Sigma$, $\scrm(\Omega)$ is the mass of all material in the body in local states contained in $\Omega$. It is possible that $\scrm(\Omega) = 0$ for a particular open set $\Omega$ in $\Sigma$, as when none of the material in the body is in a state contained in $\Omega$. By the \emph{support of} $\scrm(\cdot)$ we mean the set of states that are actually present. More precisely, if $\mu(\cdot)$ is a positive measure on $\Sigma$, then the \emph{support of} $\mu(\cdot)$ (denoted $\supp \mu(\cdot)$) is the complement in $\Sigma$ of the union of all open sets $\Omega$ such that $\mu(\Omega) = 0$.

For a Kelvin-Planck theory \theory, a \emph{Carnot element} is a cyclic reversible element $(0,\scrq)$ of $\scrPhat$, where $\scrq(\cdot)$ has a representation of the following kind: There are hotness levels $h'$ and $h$ in $\Sigma$ such that 
\begin{equation}
\scrq(\cdot) = \mu'(\cdot) - \mu(\cdot),\label{eq:Carnotq}
\end{equation}
where $\mu'(\cdot)$ and $\mu(\cdot)$ are non-zero measures in $\MSigmaPl$ satisfying
\begin{equation}\label{eq:suppCondition}
\supp\, \mu'(\cdot)\, \mathrm{\;lies\; in\,\,} h' \quad\mathrm{and}\quad \supp\, \mu(\cdot)\,  \mathrm{\;lies\; in\,\,} h.
\end{equation}
In this case, the Carnot element \emph{operates between hotness levels h' and h}.
With $\scrq(\cdot)$ in \eqref{eq:Carnotq} interpreted as a heating measure, \eqref{eq:suppCondition} requires that heat be absorbed only by material of hotness $h'$ and emitted only by material of hotness $h$.

Note that, for a Kelvin-Planck theory $\theory$, a reversible member of $\scrPhat$ of the form  $(0, c'\delta_{\sigma'} - c\delta_{\sigma})$, with   $c'$  and  $c$  positive numbers, is a Carnot element. It is a very singular one, indicating that heat is absorbed only by material precisely in state $\sigma'$ and is emitted only by material precisely in state $\sigma$. 


\section{Uniqueness of Temperature in a Kelvin-Planck Theory}\label{sec:TempUniqueness}

Suppose that, for a Kelvin-Planck theory \theory, $T^{\circ}(\cdot)$ is a Clausius-Duhem temperature scale. Then, for every positive number $\alpha$, $\alpha T^{\circ}(\cdot)$ is also a Clausius-Duhem temperature scale. In fact, for $T^{\circ}(\cdot)$ to be a Clausius-Duhem temperature scale there must be for the theory an entropy function $\eta^{\circ}(\cdot)$  such that the pair $(\eta^{\circ}(\cdot),T^{\circ}(\cdot))$ satisfies the Clausius-Duhem condition \eqref{eq:CDIneqFormal}.  It is not difficult to see that if $\alpha$ and $\beta$ are numbers, with $\alpha$ positive, then  the functions $\eta(\cdot)$ and $T(\cdot)$ defined by 
\begin{equation} \label{eq:CDPairNonunique}
T(\cdot) = \alpha\To(\cdot)\quad and \quad \eta(\cdot) = \frac{1}{\alpha}\etao(\cdot)
 + \beta. 
\end{equation}
satisfy the condition \eqref{eq:CDIneqFormal} and so also constitute a Clausius-Duhem pair for \theory. For this reason, we cannot expect uniqueness of a Clausius-Duhem temperature scale, nor can we expect uniqueness of a Clausius-Duhem pair. 

Hereafter we say that a Clausius-Duhem temperature scale $T^{\circ}(\cdot)$ for a Kelvin-Planck theory is \emph{essentially unique} if every other Clausius-Duhem temperature scale is a positive multiple of $T^{\circ}(\cdot)$. Similarly we say that a Clausius-Duhem pair $(\eta^{\circ}(\cdot),T^{\circ}(\cdot))$ for a Kelvin-Planck theory is \emph{essentially unique} if every other Clausius-Duhem pair for the theory is of the form \eqref{eq:CDPairNonunique} for some values of $\alpha > 0$ and $\beta$.

In this section we will examine conditions that are both necessary and sufficient for a Kelvin-Planck theory to possess an essentially unique Clausius-Duhem temperature scale. In the next section, we will examine additional necessary and sufficient conditions for which the corresponding the corresponding specific entropy function is also essentially unique.

From the cartoon example in Section \ref{sec:cartoon} we learned that, in consideration of a toy two-state substance, essential uniqueness of a Clausius temperature scale required that the set of cyclic processes be sufficiently broad as to ensure that $\scrChat$ contains a full line in $\mathbb{R}^2$ through the origin and, as a consequence, a large supply of reversible elements. 

The situation here will be similar. For a Kelvin-Planck theory \theory\, to have an essentially unique Clausius-Duhem temperature scale it is not only sufficient  \emph{but also necessary} that $\scrPhat$ contain a very large supply of reversible elements, in fact the entire linear subspace of $\VSigma$ consisting of all members of $\VSigma$ of the form $(0,\scrq)$ that satisfy the Clausius-Duhem $equality$. This will imply, among other things, that $\scrPhat$ \emph{must} be very rich in Carnot elements.

\medskip
In the statement of Theorem \ref{thm:UniquenessTemp}, $\delta_{\sigma}$ and $\delta_{\sigma'}$ are the Dirac measures at $\sigma$ and $\sigma'$, as described in \S\,\ref{subsubsec:MathPrelims}.

\begin{theorem}[Temperature Uniqueness and the Inexorable Role of Carnot Elements]\label{thm:UniquenessTemp}  Let \theory\ be a Kelvin-Planck theory and let $T(\cdot)$ be a Clausius-Duhem temperature scale on $\Sigma$. The following are equivalent.
\begin{enumerate}[(i)]
\item Every Clausius-Duhem temperature scale on $\Sigma$ is a positive multiple of $T(\cdot)$.
\item If \scrq\ is a member of $\scrM(\Sigma)$ that satisfies
\begin{equation}\label{eq:q/TIntegrates to zero}
\int_{\Sigma}\frac{d\scrq}{T} = 0
\end{equation}
then $(0,\scrq)$ is a member of \scrPhat.  
\item For each pair of hotness levels in $\Sigma$ there is a Carnot element of \theory\  operating between them.
\item For each pair of states $\sigma' \in \Sigma$ and $\sigma \in \Sigma$ there is a Carnot element of \theory\ operating between them,  having the form $(0, \scrq)$ with
\begin{equation}\label{eq:qInTwoStateCarnotCycle}
\scrq = c' \,\delta_{\sigma'} -  c\,\delta_{\sigma} \quad \mathrm{and} \quad \frac{c'}{c} = \frac{T(\sigma')}{T(\sigma)}.
\end{equation} 
\end{enumerate}
\end{theorem}

Theorem \ref{thm:UniquenessTemp} is proved in \cite{feinberglavine2024entropy2}. (See also \cite{feinberg1983thermodynamics}.)  The proof draws heavily on the Hahn-Banach Theorem, in particular for proof of the implication $(i) \Rightarrow (ii)$. 

	Note that the implication $(iv) \Rightarrow (i)$ is easy to prove and not very deep. However, the converse $(i) \Rightarrow (iv)$, which  ultimately relies on the Hahn-Banach Theorem, tells us something far from obvious, highlighted in the following remark:
\smallskip
	
\begin{rem}[\emph{Uniqueness of a Clausius-Duhem temperature scale requires a large supply of reversible processes}]\label{rem:TempUniquenessAndEquil}
For essential \emph{uniqueness}  on $\Sigma$ of a Clausius-Duhem temperature scale,  it is  \emph{necessary} that  for every state in $\Sigma$ there is, at least in the limit, a reversible process that visits it. This will have important implications in Section \ref{sec:Conclusions} for our discussion of why some might insist that $\Sigma$ --- the common domain of the thermodynamic-temperature and specific-entropy functions --- be restricted to ``states of equilibrium."
\end{rem}
\smallskip

\section{Uniqueness of Entropy in a Kelvin-Planck Theory}	 \label{sec:EntUniqueness}

In this section we give equivalent necessary and sufficient conditions for essential uniqueness of the specific entropy function in a Kelvin-Planck theory having an essentially unique Clausius-Duhem temperature scale. The following theorem is  proved in \cite{feinberglavine2024entropy2}. (For statement of a similar theorem see \cite{feinberg1986foundations}.)
	 
\begin{theorem}[Entropy function uniqueness and the inexorable role of reversibility] \label{thm:EntUniqueness}
 For a Kelvin-Planck theory \theory, let $\To(\cdot)$ be an essentially unique  Clausius-Duhem temperature scale (in which case all of the equivalent conditions in Theorem \ref{thm:UniquenessTemp} obtain). Suppose also that $\etao(\cdot)$ is a Clausius-Duhem specific-entropy function corresponding to  $\To(\cdot)$.   The following are equivalent:
\begin{enumerate}[(i)]
\item If $\eta\,(\cdot)$ is any other Clausius-Duhem  specific-entropy function corresponding to $\To(\cdot)$, then $\eta\,(\cdot)$ differs from $\etao(\cdot)$ by at most a constant.\smallskip

\item \scrPhat\  contains the hyperplane consisting of all  $(\scrv,\scrw)$ in $\VSigma$ such that
\begin{equation} \label{eq:hyperplane}
 \int_{\Sigma}\etao\, d\scrv = \int_{\Sigma}\frac{d\scrw}{
\To}.
\end{equation}

\item For each choice of distinct $\sigma'$ and $\sigma$ in $\Sigma$, \scrPhat\   contains a reversible element of the form $(\delta_{\sigma'} - \delta_{\sigma}, \scrq)$ for some $\scrq$ in $\MSigma$. 
\end{enumerate}
\end{theorem}
\medskip

Hereafter, we will say that, in a thermodynamic theory \theory, two states $\sigma'$ and $\sigma$ in $\Sigma$ are \emph{reversibly related} if there is in $\scrPhat$ a reversible element of the form $(\delta_{\sigma'} - \delta_{\sigma}, \scrq)$ for some $\scrq$ in $\MSigma$. Viewed as a process, it would be a reversible one with heating measure $\scrq$ and change of condition $\deltam = \delta_{\sigma'} - \delta_{\sigma}$, which is suggestive of a change in which the body experiencing the process begins will all material in state $\sigma$ and terminates will all material in state $\sigma'$.

The implication $(iii) \Rightarrow (i)$ is easy to prove. The Hahn-Banach Theorem enters the proof of the far deeper converse, highlighted in the following remark.
\smallskip

\begin{rem}[\emph{Uniqueness and reversibility, again}] \label{rem:EntUniquenessReqReversibility}
 For a Kelvin-Planck theory $\theory$ having an essentially unique Clausius-Duhem temperature scale, entropy-density-function \emph{uniqueness} on $\Sigma$ \emph{requires} that \emph{all} states in $\Sigma$ be reversibly related. This requirement, like the implication $(i) \Rightarrow (iv)$ in Theorem \ref{thm:UniquenessTemp},  provides something of a motivation for those who might insist that  $\Sigma$ be restricted to equilibrium states. These implications will be discussed in Section \ref{sec:Conclusions}.
\end{rem}
\smallskip

\begin{rem}[\emph{Uniqueness of entropy and temperature requires that $\scrPhat$ be either a hyperplane or a half-space}] In Section \ref{sec:cartoon}, we discovered that, for our cartoon two-state substance, uniqueness of a Clausius temperature scale occurs in precisely two circumstances: $\scrChat$ is either a hyperplane or a half-space in 
$\R^2$. In a more general setting, the situation is similar. If, for a Kelvin-Planck theory \theory, $(\eta^{\circ}(\cdot), T^{\circ}(\cdot))$ is an essentially unique Clausius-Duhem pair, then $\scrPhat$ is either the hyperplane consisting of all $(\scrv,\scrw)$ in $\VSigma$ satisfying \eqref{eq:hyperplane} or else it is the half-space consisting of all $(\scrv,\scrw)$ in $\VSigma$ satisfying
\begin{equation} \label{eq:halfspace}
 \int_{\Sigma}\etao\, d\scrv \geq \int_{\Sigma}\frac{d\scrw}{
\To}.
\end{equation}
In particular, if $\scrPhat$ contains \emph{even one} irreversible element, then $\scrPhat$ must contain \emph{all} $(\scrv,\scrw)$ in $\VSigma$ satisfying \eqref{eq:halfspace}. See \cite{feinberglavine2024entropy2}.
\end{rem}

\section{Concluding Remarks: The Legacy of the 19th Century Pioneers in Light of 20th Century Mathematics} \label{sec:Conclusions}

As we asserted at the outset, the 19\textsuperscript{th} century pioneers of classical thermodynamics were brilliant in what they gave us, but they did not have available to them the 20\textsuperscript{th}  century mathematical tools that make possible a deeper analysis or a broadening of what they had done. In this 21\textsuperscript{st} century review, we hope that Theorems \ref{thm:CDPairExistence}, \ref{thm:UniquenessTemp}, \ref{thm:EntUniqueness}, and the commentary below will put into a modern perspective the miraculous path that the great pioneers traveled.
%
%

\begin{subsection}{Conflation of Existence and Uniqueness, and of  Necessary and Sufficient Conditions} \label{subsec:Conflation} In the hands of the pioneers,  the argument for \emph{existence} of 
 a Clausius-Duhem entropy-temperature pair was a constructive one. The construction relied on the presence of a very large supply of reversible processes (in particular Carnot cycles), wherein \emph{each} state in the common domain of the entropy and temperature state-functions is visited by at least one such reversible process. For the pioneers, \emph{uniqueness} of such a Clausius-Duhem pair, \emph{so constructed}, emerged almost simultaneously from the same argument.\footnote{Fermi \cite{fermithermodynamics} gives a good summary.} It is important to note that reversibility of a process was deemed, intuitively, to derive from the extreme slowness of its execution, so that the material states visited by reversible processes were regarded to be ones of equilibrium.
 
 \begin{rem} Although the association of a  reversible process with a slow passage through a trajectory of equilibria is compelling, it is reasonable to draw a distinction between two ideas: an equilibrium state and a state that can manifest itself in a reversible process. In Appendix C of \cite{feinberglavine2024entropy2} we gave a hypothetical reacting-mixture example  indicating that, in the sense of \S\ref{subsec:Reversible&Cyc},  reversible elements of \scrPhat\  can derive from consideration of fast processes, involving mixture states far from chemical equilibrium.
 \end{rem}

\begin{subsubsection}{Summary: Existence of a Clausius-Duhem pair} \label{subsubsec:SummaryExist} Because, in the classical arguments, the states to which the temperature and entropy functions assign values were deemed to be equilibria, the assertion is often made (as in the Gemini response in $\S$\ref{subsec:Motivation-Questions}) that it is \emph{necessary} that the domains of those functions be restricted to equilibrium states. From the standpoint of cautious scholarship, there is indeed merit to this assertion, \emph{if the assertion is taken to mean that the equilibrium-state requirement is necessary for the 19\textsuperscript{\ th} century argument to proceed in the way that it did.}

However, in our terms, \emph{this is not the same as an assertion that, for existence of a Clausius-Duhem pair in a Kelvin-Planck theory $\theory$, it is necessary that all members of $\Sigma$ be states of equilibrium (however equilibrium might be defined).} Such an assertion would be false. Theorem \ref{thm:CDPairExistence} tells us that, for \emph{existence} of a Clausius-Duhem entropy-temperature pair in a thermodynamic theory, the Kelvin-Planck Second Law,  as encoded in \eqref{eq:KPFinal}, is \emph{by itself sufficient}. As indicated in the proof of $(i) \Rightarrow (ii)$ in Theorem \ref{thm:CDPairExistence},  this is an almost immediate consequence of the Hahn-Banach Theorem. Additional suppositions about equilibrium or reversibility of processes are \emph{not necessary}.\footnote{It should be kept in mind, however, that $\Sigma$ is presumed to be a closed and bounded subset of $\mathbb{R}^N$, for $N$ finite. See also $\S$\ref{subsec:WhereThingsStand}.}

\end{subsubsection}

\begin{subsubsection}{Summary: Uniqueness of a Clausius-Duhem pair} \label{subsubsec:SummaryUnique}
Essential \emph{uniqueness} of a Clausius-Duhem pair is  very different. Theorems \ref{thm:UniquenessTemp} and \ref{thm:EntUniqueness} tell us that, for a Kelvin-Planck theory $\theory$ to have a Clausius-Duhem entropy-temperature pair that is \emph{essentially unique on the entirety of} $\Sigma$, it is \emph{necessary} that \emph{every} state in $\Sigma$ be visited by a reversible element in $\scrPhat$. This too is a consequence of the Hahn-Banach Theorem.\footnote{Lack of uniqueness of a Clausius-Duhem pair without a full supply of reversible processes was already discussed with respect to temperature in \cite{feinberg1983thermodynamics} and with respect to both temperature and entropy in \cite{feinberg1986foundations}. Later, in a very different setting, Lieb and Yngvason \cite{lieb2013entropy} also discussed lack of uniqueness of entropy when nonequilibrium states are taken into account.} 

	Nevertheless, even when a Kelvin-Planck theory \theory\, admits a variety of essentially different Clausius-Duhem pairs, the component temperature and entropy functions \emph{for all such pairs} will often have properties expected in applications. For example, it is shown in \cite{feinberglavine2024entropy2} that \emph{every} Clausius-Duhem temperature scale for $\theory$ will reflect  a natural ``hotter than" relation in the states of $\Sigma$, a relation that makes no mention of temperature but is instead defined solely in terms of the processes the theory contains.

	Note that Theorems \ref{thm:UniquenessTemp} and \ref{thm:EntUniqueness} also give \emph{sufficient} conditions for essential uniqueness of a Clausius-Duhem temperature scale and, when such conditions are satisfied, essential uniqueness of a corresponding entropy-temperature Clausius-Duhem pair.	
	
\medskip
\begin{rem}[\emph{Uniqueness of a Clausius-Duhem entropy-temperature pair on a state space subdomain}]\label{rem:CDPairUniqueOnSubdomain}

For a Kelvin-Planck theory $\theory$ it might happen that all Clausius-Duhem pairs are essentially the same on a subdomain $\Sigma_0$ within $\Sigma$. In that case, it is only states within $\Sigma_0$ that need be visited by reversible members of $\scrPhat$.  See  \cite{feinberg1983thermodynamics}  and \cite{feinberglavine2024entropy2}.  This allows for a Kelvin-Planck theory in which $\Sigma$ is a mix of states that are visited by reversible members of $\scrPhat$ and those that are not.	
\end{rem} 
\end{subsubsection}
\end{subsection}

\begin{subsection}{Use of the Clausius-Duhem Inequality in Modern Continuum Thermomechanics}

Here we return to Section \ref{subsec:Motivation-Questions}, which posed two questions.  Section \ref{subsec:Conflation} provides answers to those questions, at least for a thermodynamic theory in which the state space can be identified with a closed and bounded subset of $\mathbb{R}^N$, with $N$ finite. 

However, an overarching question in Section \ref{subsec:Motivation-Questions} was about the free-wheeling use of the Clausius-Duhem inequality in modern continuum thermomechanics when rapid deformations and great heat fluxes are present.\footnote{The general Clausius-Duhem inequality stated here as \eqref{eq:CDIneqFormal} is stated in a second version as eq.(5.9) in \cite{feinberglavine2024entropy1}. That version, stated explicitly in terms of bodies and material points, makes more tangible contact with the continuum thermomechanics literature.}   Two examples cited were an article  by Coleman and Noll \cite{coleman1963thermodynamics}, largely about viscous fluids, and an article by Bowen \cite{bowen1968thermochemistry}, in which the Coleman-Noll methodology was applied to mixtures with chemical reactions. In both instances, the relevant state spaces can be reasonably associated with closed and bounded subsets of $\mathbb{R}^N$, with $N$ finite.

In \cite{coleman1963thermodynamics} and \cite{bowen1968thermochemistry} the \emph{existence} of a Clausius-Duhem entropy-temperature function pair is taken for granted at the outset (along with the differentiability almost always presumed tacitly in the thermodynamics literature\footnote{The central theorems here are focused on existence and properties of \emph{continuous} entropy and temperature functions of state. However, the same theorems obtain when \emph{continuous} is replaced by \emph{differentiable} (of various degrees), provided that the idea of an open set in $\VSigma$ is modified, along lines discussed in Remark 10.2 of \cite{feinberg1983thermodynamics}.}).  For the authors, the Second Law then amounts to the following assertion: For no application of external field forces or external radiative heat supplies to a body can any process experienced by the body, consistent with the remaining laws of physics, violate the Clausius-Duhem inequality.  

In both \cite{coleman1963thermodynamics} and \cite{bowen1968thermochemistry}, application of this principle yields a family of identities and inequalities that hold among various \textit{constitutive functions} characterizing the material under study. These are functions of the local state that give values for local physical attributes such as the specific internal energy, the conductive heat flux, the stress tensor, and  chemical reaction rates. 

In the spirit of Section \ref{subsec:Motivation-Questions} we can now ask again:  For the purposes of such applications, does the presumed \emph{existence} of specific-entropy and temperature functions of the local (perhaps nonequilibrium) state, suited to the Clausius-Duhem inequality, follow from a more fundamental statement of the Second Law?

We believe that the answer is \emph{yes}. Theorem \ref{thm:CDPairExistence} tells us that existence of a Clausius-Duhem pair is not only ensured by the Kelvin-Planck Second Law, as embodied in \eqref{eq:KPFinal}, but also equivalent to it.

	However,  Theorems \ref{thm:UniquenessTemp} and \ref{thm:EntUniqueness} tell us with certainty that, for a Kelvin-Planck theory $\theory$, not only \emph{might not} a Clausius-Duhem  pair be essentially unique on the entire state space $\Sigma$, such uniqueness \emph{ cannot} obtain unless \emph{every} state in $\Sigma$ is visited by a reversible element of $\scrPhat$. 
	
	In this light, \emph{the question of interest in nonequilibrium contexts shifts from one about existence of a Clausius-Duhem entropy-temperature pair to one about how important its  uniqueness is in a particular application}. That importance will certainly vary from one application to another and must ultimately be judged against the aims of the applier. It cannot be decided here.
	
\begin{rem}[Considerations about uniqueness] In \cite{coleman1963thermodynamics} and \cite{bowen1968thermochemistry}, the authors begin with the presumed existence of a \emph{fixed} Clausius-Duhem pair, which from our standpoint might be one among many. Because in those articles all identities derive  from existence of a fixed but \emph{arbitrary} Clausius-Duhem pair, those same identities (which sometimes involve the entropy and temperature explicitly) would presumably hold for \emph{every} choice of a Clausius-Duhem pair. However, for some of the results stated, such as the positivity of the viscosity and thermal conductivity in \cite{coleman1963thermodynamics}, entropy and temperature play no overt role.
\end{rem}

\begin{rem}[Materials with memory] Our focus in this article has been on thermodynamic theories in which the state space can be identified with a closed and bounded subset of $\mathbb{R}^N$, with $N$ finite. There are other instances, however, in which models of material behavior transcend the range of applicability of theorems stated here. We have in mind  theories of materials with memory, such as \cite{coleman1964thermodynamics}, in which the state of a material element is associated with a list of \emph{histories} of material attributes, such as the deformation gradient and the temperature. In this case, the state space would reside in an infinite-dimensional vector space, with members consisting of functions of time.  
\end{rem}

\end{subsection}

\begin{subsection}{Where Do Things Stand?} \label{subsec:WhereThingsStand}

In Section \ref{subsec:Motivation-Questions}, a  2026 Google Gemini response was given for the  query: \textit{``In classical thermodynamics, is entropy defined only on equilibrium states? If so, why?"} That answer, with some elaboration, was \emph{yes}. We believe that such a consensus answer, and the elaboration, would be the same were it given by a hypothetical artificial intelligence machine of the 19\textsuperscript{th} century. We also think that the two responses would  be identical if, in the query, ``entropy" were replaced by ``thermodynamic temperature." It seems, then, that consensus 21\textsuperscript{st} century thinking about classical thermodynamics remains largely unchanged since the time of the 19\textsuperscript{th} century pioneers.

Especially with regard to \emph{existence} of local entropy and thermodynamic temperature functions of state suited to the Clausius-Duhem inequality for rapid processes, the Hahn-Banach existence proofs described in this review were known 40 years ago \cite{feinberg1983thermodynamics,feinberg1986foundations}. The fact that the consensus view of classical thermodynamics nevertheless remains largely that of the 19\textsuperscript{th} century can be attributed to several factors. Some certainly involve the dramatic evolution of 20\textsuperscript{th} century physics, in particular the drift in physics education and research away from classical thermodynamics and toward statistical thermodynamics.

	Perhaps the most important factor, however, is the great distance between the mathematics education normally provided to most practitioners of classical thermodynamics (e.g., engineers, chemists, biologists) and the mathematics education required for understanding of the Hahn-Banach \emph{arguments} underlying the theorems reported here (as distinct from their \emph{results}). It is our hope that this review will help narrow that gap.

\end{subsection}

\bibliographystyle{RS} 

\bibliography{MonoLibrary-Thermo}

\end{document}